\documentclass[12pt,draftcls,onecolumn]{IEEEtran}

\usepackage{amsmath}
\usepackage{amssymb}  
\usepackage{amsfonts}
\usepackage{mathtools}
\usepackage{bm}
\usepackage{color}
\usepackage{graphics} 
\usepackage{graphicx}
\usepackage{epsfig} 
\usepackage{epstopdf}
\usepackage[noadjust]{cite}
\usepackage{tikz}
\usetikzlibrary{calc,positioning,shapes,shadows,arrows,fit}
\usepackage{fancyhdr}
\usepackage{fancyref}
\usepackage{hyperref}
\usepackage{float}
\usepackage[font = footnotesize]{caption}

\usepackage{amsthm}

\newtheorem{theorem}{Theorem}
\newtheorem{lemma}{Lemma}

\theoremstyle{definition}

\newtheorem{definition}{Definition}
\newtheorem{remark}{Remark}
\newtheorem{assumption}{Assumption}

\newtheorem{problem}{Problem}
\newtheorem{property}{Property}

\begin{document}

\title{Decentralized Control of Uncertain Multi-Agent Systems with Connectivity Maintenance and Collision Avoidance}

\author{Alexandros Filotheou, Alexandros Nikou and Dimos V. Dimarogonas
\thanks{The authors are with the ACCESS Linnaeus Center, School of Electrical
	Engineering, KTH Royal Institute of Technology, SE-100 44, Stockholm,
	Sweden and with the KTH Center for Autonomous Systems. Email: {\tt\small \{alefil, anikou, dimos\}@kth.se}. This work was supported by the H2020 ERC Starting Grant BUCOPHSYS, the EU H2020 Co4Robots project, the Swedish Foundation for Strategic Research (SSF), the Swedish Research Council (VR) and the Knut och Alice Wallenberg Foundation (KAW).}}



\maketitle

\begin{abstract}
This paper addresses the problem of navigation control of a general class of uncertain nonlinear multi-agent systems in a bounded workspace of $\mathbb{R}^n$ with static obstacles. In particular, we propose a decentralized control protocol such that each agent reaches a predefined position at the workspace, while using only local information based on a limited sensing radius. The proposed scheme guarantees that the initially connected agents remain always connected. In addition, by introducing certain distance constraints, we guarantee inter-agent collision avoidance, as well as, collision avoidance with the obstacles and the boundary of the workspace. The proposed controllers employ a class of Decentralized Nonlinear Model Predictive Controllers (DNMPC) under the presence of disturbances and uncertainties. Finally, simulation results verify the validity of the proposed framework.
\end{abstract}

\begin{IEEEkeywords}
Multi-Agent Systems, Cooperative control, Decentralized Control, Robust Control, Nonlinear Model Predictive Control, Collision Avoidance, Connectivity Maintenance.
\end{IEEEkeywords}

%
\IEEEpeerreviewmaketitle

\section{Introduction}

During the last decades, \emph{decentralized control of multi-agent systems} has gained a significant amount of attention due to the great variety of its applications, including  multi-robot systems, transportation, multi-point surveillance and biological systems. An important topic of research is \emph{multi-agent navigation} in both the robotics and the control communities, due to the need for autonomous control of multiple robotic agents in the same workspace. Important applications of multi-agent navigation arise also in the fields of air-traffic management and in autonomous driving by guaranteeing collision avoidance with other vehicles and obstacles. Other applications are formation control, in which the agents are required to reach a predefined geometrical shape (see e.g., \cite{alex_chris_ppc_formation_ifac}) and high-level planning where it is required to provide decentralized controllers for navigating the agents between regions of interest of the workspace (see e.g., \cite{alex_acc_2016, alex_automatica_2017}).

The literature on the problem of navigation of multi-agent systems is rich. In \cite{dimos_2006_automatica_nf}, (\cite{makarem_decentralized_nf}), a decentralized control protocol of multiple non-point agents (point masses) with collision avoidance guarantees is considered. The problem is approached by designing navigation functions which have been initially introduced in \cite{koditschek1990robot}. However, this method requires preposterously large actuation forces and it may give rise to numerical instability due to computations of exponentials and derivatives. A decentralized potential field approach of navigation of multiple unicycles and aerial vehicles with collision avoidance has been considered in \cite{panagou_potential_fields_unicycle} and \cite{baras_decentralized_control}, respectively; Robustness analysis and saturation in control inputs are not addressed. In \cite{roozbehani2009hamilton}, the collision avoidance problem for multiple agents in intersections has been studied. An optimal control problem is solved, with only time and energy constraints. Authors in \cite{loizou_2017_navigation_transformation} proposed decentralized controllers for multi-agent navigation and collision avoidance with arbitrarily shaped obstacles in $2$D environments. However, connectivity maintenance properties are not taken into consideration in all the aforementioned works. 

Other approaches in multi-agent navigation propose solutions to distributed optimization problems. In \cite{Dunbar2006549}, a decentralized receding horizon protocol for formation control of linear multi-agent systems is proposed. The authors in \cite{aguiar_multiple_uavs_linearized_mpc} considered the path-following problems for multiple Unmanned Aerial Vehicles (UAVs)  in which a distributed optimization method is proposed through linearization of the dynamics of the UAVs. A DNMPC along with potential functions for collision avoidance has been studied in \cite{distributed_mpc_sastry}. A feedback linearization framework along with Model Predictive Controllers (MPC) for multiple unicycles in leader-follower networks for ensuring collision avoidance and formation is introduced in \cite{fukushima_2005_distributed_mpc_linearization}. The authors of \cite{kevicky1, kevicky2, 4459797} proposed a decentralized receding horizon approach for discrete time multi-agent cooperative control.  However, in the aforementioned works, plant-model mismatch or uncertainties and/or connectivity maintenance are not considered. In \cite{1383977} and \cite{1429425} a centralized and a decentralized linear MPC formulation and integer programming is proposed, respectively, for dealing with collision avoidance of multiple UAVs.

The main contribution of this paper is to propose a novel solution to a general navigation problem of uncertain multi-agent systems, with limited sensing capabilities in the presence of bounded disturbances and bounded control inputs. We propose a Decentralized Nonlinear Model Predictive Control (DNMPC) framework in which each agent solves its own optimal control problem, having only availability of information on the current and planned actions of all agents within its sensing range. In addition, the proposed control scheme, under relatively standard Nonlinear Model Predictive Control (NMPC) assumptions, guarantees connectivity maintenance between the agents that are initially connected, collision avoidance between the agents, collision avoidance between agents and static obstacles of the environment, from all feasible initial conditions. To the best of the authors' knowledge, this is the first time that this novel multi-agent navigation problem is addressed. This paper constitutes a generalization of a submitted work \cite{IJC_2017} in which the problem of decentralized control of multiple rigid bodies under Langrangian dynamics in $\mathbb{R}^3$ is addressed.

The remainder of this paper is structured as follows: In Section \ref{sec:notation_preliminaries} the notation and preliminaries are given. Section \ref{sec:problem_formulation} provides the system dynamics and the formal problem statement. Section \ref{sec:main_results} discusses the technical details of the solution and Section \ref{sec:simulation_results} is devoted to simulation examples. Finally, conclusions and future work are discussed in Section \ref{sec:conclusions}.

\section{Notation and Preliminaries} \label{sec:notation_preliminaries}

The set of positive integers is denoted by $\mathbb{N}$. The real $n$-coordinate
space, $n\in\mathbb{N}$, is denoted by $\mathbb{R}^n$;
$\mathbb{R}^n_{\geq 0}$ and $\mathbb{R}^n_{> 0}$ are the sets of real
$n$-vectors with all elements nonnegative and positive, respectively. The notation
$\|x\|$ is used for the Euclidean norm of a vector $x \in \mathbb{R}^n$ and $\|A\| = \max \{ \|A x \|: \|x\| = 1 \}$ for the induced norm of a matrix $A \in \mathbb{R}^{m \times n}$. Given a real symmetric matrix $A$, $\lambda_{\text{min}}(A)$
and $\lambda_{\text{max}}(A)$ denote the minimum and the maximum absolute value of eigenvalues of $A$, respectively.
Its minimum and maximum singular values are denoted by $\sigma_{\text{min}}(A)$ and $\sigma_{\text{max}}(A)$ respectively; $I_n \in \mathbb{R}^{n \times n}$ and $0_{m \times n} \in \mathbb{R}^{m \times n}$ are the identity matrix and the $m \times n$ matrix with all entries zeros,
respectively. The set-valued function $\mathcal{B}:\mathbb{R}^n\times\mathbb{R}_{> 0} \rightrightarrows \mathbb{R}^n$, given as $\mathcal{B}(x, r) = \{y \in \mathbb{R}^n: \|y-x\| \leq r\},$ represents the $n$-th dimensional ball with center $x \in \mathbb{R}^{n}$ and radius $r \in \mathbb{R}_{> 0}$. 

\begin{definition} \label{def:p_difference}
	Given the sets $\mathcal{S}_1$, $\mathcal{S}_2 \subseteq \mathbb{R}^n$, the \emph{Minkowski addition} and the \emph{Pontryagin difference} are defined by: $\mathcal{S}_1 \oplus \mathcal{S}_2 = \{s_1 + s_2 \in \mathbb{R}^n : s_1 \in \mathcal{S}_1, s_2 \in \mathcal{S}_2\}$ and $\mathcal{S}_1 \ominus \mathcal{S}_2 = \{s_1 \in \mathbb{R}^n: s_1+s_2 \in \mathcal{S}_1, \forall \ s_2 \in \mathcal{S}_2\}$, respectively.
\end{definition}

\begin{property} \label{property:set_property}
	Let the sets $\mathcal{S}_1$, $\mathcal{S}_2$, $\mathcal{S}_3 \subseteq \mathbb{R}^n$. Then, it holds that:
	\begin{align}
	(\mathcal{S}_1 \ominus \mathcal{S}_2) \oplus (\mathcal{S}_2 \ominus \mathcal{S}_3) = (\mathcal{S}_1 \oplus \mathcal{S}_2) \ominus (\mathcal{S}_3 \oplus \mathcal{S}_3). \label{eq:a_minus_b_plus_c}
	\end{align}
\end{property}
\begin{proof}
	The proof can be found in Appendix \ref{app:proof_set_property}.
\end{proof}

\begin{definition} \cite{khalil_nonlinear_systems} \label{def:k_class}
A continuous function $\alpha : [0, a) \to \mathbb{R}_{\ge 0} $ belongs to \emph{class $\mathcal{K}$} if it is strictly increasing and $\alpha (0) = 0$. If $a = \infty$ and $\lim\limits_{r \to \infty} \alpha(r) = \infty$, then function $\alpha$ belongs to class $\mathcal{K}_{\infty}$. A continuous function $\beta : [0, a) \times \mathbb{R}_{\ge 0} \to \mathbb{R}_{\ge 0}$ belongs to \emph{class $\mathcal{KL}$} if: $1)$ for a fixed $s$, the mapping $\beta(r,s)$ belongs to class $\mathcal{K}$ with respect to $r$; $2)$ for a fixed $r$, the mapping $\beta(r,s)$ is strictly decreasing with respect to $s$, and it holds that: $\lim\limits_{s \to \infty} \beta(r,s) = 0$.
\end{definition}

\begin{definition} \cite{marquez2003nonlinear}
	\label{def:ISS}
	A nonlinear system $\dot{x} = f(x,u)$, $x \in \mathcal{X}$, $u \in \mathcal{U}$ with initial condition $x(0)$ is said
	to be \emph{Input-to-State Stable (ISS)} in $\mathcal{X}$ if there exist functions $\sigma \in \mathcal{K}$ and $\beta \in \mathcal{KL}$ and constants $k_1$, $k_2 \in \mathbb{R}_{> 0}$ such that: $$\|x(t)\| \leq \beta\big(\|x(0)\|,t\big) + \sigma \left(\displaystyle \sup_{t \in \mathbb{R}_{\ge 0}} \|u(t)\|\right), \forall t \geq 0,$$ for all $x(0) \in \mathcal{X}$ and $u \in \mathcal{U}$ satisfying $\|x(0)\| \leq k_1$ and $\displaystyle \sup_{t \in \mathbb{R}_{\ge 0}} \|u(t)\| \leq k_2$.
\end{definition}

\begin{definition} \cite{marquez2003nonlinear}
	\label{def:ISS_Lyapunov}
	A continuously differentiable function $V: \mathcal{X} \to \mathbb{R}_{\geq 0}$ for the nonlinear system $\dot{x} = f(x, w)$, with $x \in \mathcal{X}$, $w \in \mathcal{W}$ is an \emph{ISS Lyapunov function} in $\mathcal{X}$ if there exist class $\mathcal{K}$ functions
	$\alpha_1$, $\alpha_2$, $\alpha_3$ and $\sigma$, such that: $$\alpha_1\big(\|x\|\big) \leq V\big(x\big) \leq \alpha_2\big(\|x\|\big), \forall \ x \in \mathcal{X},$$ and $$\dfrac{d}{dt} \left[ V\big(x\big) \right] \leq \sigma\big(\|w\|\big) - \alpha_3\big(\|x\|\big), \forall \ x \in \mathcal{X}, w \in \mathcal{W}.$$
\end{definition}

\begin{theorem} \cite{Sontag2008}
	\label{def:ISS_Lyapunov_admit_theorem} 
	A nonlinear system $\dot{x} = f(x,u)$ with $x \in \mathcal{X}$, $u \in \mathcal{U}$ is \emph{Input-to-State Stable} in $\mathcal{X}$ if and only if it admits an ISS Lyapunov function in $\mathcal{X}$.
\end{theorem}

\begin{definition} \cite{khalil_nonlinear_systems}
	\label{def:positively_invariant}	
	Consider a system $\dot{x} = f(x)$, $x \in \mathbb{R}^n$. Let $x(t)$ be a solution of the system with initial condition $x_0$. Then, $\mathcal{S} \subseteq \mathbb{R}^n$ is called \emph{positively invariant set} of the system if, for any $x_0 \in \mathcal{S}$ we have $x(t) \in \mathcal{S}$, $t \in \mathbb{R}_{\ge 0}$, along every solution $x(t)$.
\end{definition}

\section{Problem Formulation} \label{sec:problem_formulation}

\subsection{System Model}

Consider a set $\mathcal{V}$ of $N$ agents, $\mathcal{V} = \{ 1,2, \ldots, N\}$, $N \geq 2$, operating in a workspace $\mathcal{D} \subseteq \mathbb{R}^n$. The workspace is assumed to be modeled by a bounded ball $\mathcal{B}\left(x_{\scriptscriptstyle \mathcal{D}},r_{\scriptscriptstyle \mathcal{D}}\right)$, where $x_{\scriptscriptstyle \mathcal{D}} \in \mathbb{R}^n$ and $r_{\scriptscriptstyle \mathcal{D}} \in \mathbb{R}_{>0}$ are its center and radius, respectively.

We consider that over time $t$ each agent $i \in \mathcal{V}$ occupies the ball $\mathcal{B}\left(x_i(t), r_i\right)$, where
	$x_i : \mathbb{R}_{\geq 0} \to \mathbb{R}^n$ is the position of the agent at time $t \in \mathbb{R}_{\ge 0}$, and $r_i < r_\mathcal{D}$ is the radius of the
agent's rigid body. The \emph{uncertain nonlinear dynamics} of each agent $i \in \mathcal{V}$ are given by:
	\begin{align}
	\dot{x}_i(t) = f_i(x_i(t), u_i(t)) + w_i(x_i(t), t), \label{eq:system}
	\end{align}
	where $u_i: \mathbb{R}_{\ge 0} \to \mathbb{R}^m$ stands for the control input of each agent and $f_i: \mathbb{R}^n \times \mathbb{R}^m \to \mathbb{R}^n$ is a twice continuously differentiable vector field satisfying $f_i(0_{n \times 1}, 0_{m \times 1}) = 0_{n \times 1}$. The continuous function $w_i: \mathbb{R}^n \times \mathbb{R}_{\ge 0} \to \mathbb{R}^{n}$ is a term representing \emph{disturbances} and \emph{modeling uncertainties}. We consider bounded inputs and disturbances as $u_i \in \mathcal{U}_i$ and $w_i \in \mathcal{W}_i$, where $\mathcal{U}_i = \{u_i \in \mathbb{R}^m : \|u_i\| \leq \widetilde{u}_i\}$ and $\mathcal{W}_i = \{ w_i \in \mathbb{R}^n : \|w_i\| \leq \widetilde{w}_i\}$, for given finite constants $\widetilde{w}_i$, $\widetilde{u}_i \in \mathbb{R}_{> 0}$, $i \in \mathcal{V}$.
\begin{assumption}
	\label{ass:g_i_g_R_Lipschitz}
	The nonlinear functions $f_i$, $i \in \mathcal{V}$ are \emph{Lipschitz continuous} in $\mathbb{R}^n \times \mathcal{U}_i$ with Lipschitz
	constants $L_{f_i}$. Thus, it holds that:
	\begin{align*}
	\| f_i(x, u) - f_i(x', u) \| \le L_{f_i} \| x - x' \|, \forall \ x, x' \in \mathbb{R}^n, u \in \mathcal{U}_i.
	\end{align*}
\end{assumption}

We consider that in the given workspace there exist $L \in \mathbb{N}$ \emph{static obstacles}, with $\mathcal{L} = \{1, 2, \dots, L\}$, also modeled by the balls $\mathcal{B}\left(x_{\scriptscriptstyle O_\ell}, r_{\scriptscriptstyle O_\ell}\right)$, with centers at positions $x_{\scriptscriptstyle O_\ell} \in \mathbb{R}^n$ and radii $r_{\scriptscriptstyle O_\ell}\in \mathbb{R}_{> 0}$, where $\ell \in \mathcal{L}$. Their positions and sizes are assumed to be known a priori to each agent.
\begin{assumption} \label{ass:measurements_access}
	Agent $i \in\mathcal{V}$ has: $1)$ access to measurements $x_i(t)$ for every $t \in \mathbb{R}_{\ge 0}$; $2)$ A limited sensing range $d_i \in \mathbb{R}_{>0}$ such that: $$d_i > \max_{i,j \in \mathcal{V}, i \neq j, \ell \in \mathcal{L}}\{r_i + r_j, r_i + r_{\scriptscriptstyle O_\ell}\}.$$			
\end{assumption}
The latter implies that each agent is capable of perceiving all other agents and all workspace obstacles. The consequence of points 1 and 2 of Assumption \ref{ass:measurements_access} is that by defining the set of agents $j$ that are within the sensing range of agent $i$ at time $t$ as: $$\mathcal{R}_i(t) \triangleq \{j\in\mathcal{V} \backslash \{i\} : \| x_i(t) - x_j(t) \| < d_i\},$$ agent $i$ is also able to measure at each time instant $t$ the vectors $x_{j}(t)$ of all agents $j \in \mathcal{R}_i(t)$. 

\begin{definition} \label{definition:collision_free_conf} The multi-agent system is in a \emph{collision-free configuration} at a time instant $\tau \in \mathbb{R}_{\ge 0}$ if the following hold:
\begin{enumerate}
\item For every $i$, $j \in \mathcal{V}$, $i \neq j$ it holds that: $\| x_i(\tau) - x_j(\tau) \| > r_{i} + r_{j}$;
\item For every $i \in \mathcal{V}$ and for every $\ell \in \mathcal{L}$ it holds that: $\|x_i(\tau) - x_{\scriptscriptstyle O_\ell}\| > r_{i} + r_{\scriptscriptstyle O_\ell}$;
\item For every $i \in \mathcal{V}$ it holds that: $\|x_{\scriptscriptstyle \mathcal{D}}-x_i(\tau)\| < r_{\scriptscriptstyle \mathcal{D}} - r_i$.
\end{enumerate}	
\end{definition}

\begin{definition}
\noindent The \emph{neighboring set} of agent $i \in \mathcal{V}$ is defined by: $$\mathcal{N}_i = \{j \in \mathcal{V} \backslash \{i\} : j \in \mathcal{R}_i(0) \}.$$ We will refer to agents $j \in \mathcal{N}_i$ as the \textit{neighbors} of agent
$i \in \mathcal{V}$.
\end{definition}

The set $\mathcal{N}_i$ is composed of indices of agents
$j \in \mathcal{V}$ which are within the sensing range of agent $i$ at time
$t=0$. Agents $j \in \mathcal{N}_i$ are agents which agent $i$ is instructed
to keep within its sensing range at all times $t \in \mathbb{R}_{>0}$, and therefore
maintain connectivity with them. While the sets $\mathcal{N}_i$ are introduced
for connectivity maintenance specifications and they are fixed, the sets $\mathcal{R}_i(t)$ are used to ensure collision avoidance, and, in general,
their composition varies through through time.

\begin{assumption} \label{ass:initial_conditions}
	For sake of cooperation needs, we assume that $\mathcal{N}_i \neq \emptyset$, $\forall i \in \mathcal{V}$, i.e., all agents have at least one neighbor. We also assume that at time $t=0$ the multi-agent system is in a \textit{collision-free configuration}, as given in Definition \ref{definition:collision_free_conf}.
\end{assumption}

\subsection{Objectives}

Given the aforementioned modeling of the system, the objective of this paper is the \textit{stabilization of the agents} $i \in \mathcal{V}$ starting
from a collision-free configuration as given in Definition \ref{definition:collision_free_conf} to a desired configuration $x_{i, \text{des}} \in \mathbb{R}^n$, while maintaining connectivity between neighboring agents, and avoiding collisions between agents, obstacles, and the workspace boundary. 

\begin{definition} \label{definition:feasible_steady_state_conf}
	The desired configuration $x_{i, \text{des}} \in \mathbb{R}^n$ of agent $i \in \mathcal{V}$ is \textit{feasible}, if the following hold: $1)$ It is a collision-free configuration according to Definition \ref{definition:collision_free_conf}; $2)$ It does not result in a violation of the connectivity maintenance constraint between neighboring agents, i.e., $\|x_{i, \text{des}} - x_{j, \text{des}}\| < d_i$, $\forall i \in \mathcal{V}, j \in \mathcal{N}_i$.
\end{definition}

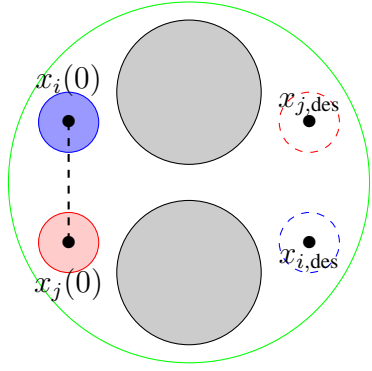
\begin{figure}[t!]
	\centering
	\begin{tikzpicture}[scale = 0.8]
	\draw [color = green] (-3.0, 0) circle (3.0cm);
	\draw [color = black, fill = black!20] (-3, 1.5) circle (1.2cm);
	\draw [color = black, fill = black!20] (-3, -1.5) circle (1.2cm);
	\draw [color = blue, fill = blue!40] (-5.0, 1.0) circle (0.5cm);
	\draw [color = red, fill = red!20] (-5.0, -1.0) circle (0.5cm);
	\draw [color = red, dashed] (-1.0, 1.0) circle (0.5cm);
	\draw [color = blue, dashed] (-1.0, -1.0) circle (0.5cm);
	
	\draw [color=black,thick,dashed,-](-5, 1) to (-5, -1);
	\node at (-5, 1) {$\bullet$};
	\node at (-5, -1) {$\bullet$};
	\node at (-1, 1) {$\bullet$};
	\node at (-1, -1) {$\bullet$};
	\node at (-5.0, -1.7) {$x_j(0)$};
	\node at (-5.0, 1.7) {$x_i(0)$};
	\node at (-1.0, -1.3) {$x_{i, \text{des}}$};
	\node at (-1.0, 1.3) {$x_{j, \text{des}}$};
	
	\end{tikzpicture}
	\caption{An example of infeasible initial conditions. Consider two agents $i$, $j$, depicted with blue and red, respectively, in a workspace (depicted by green) with two obstacles (depicted by gray). It is required to design controllers $u_i$, $u_j$ that navigate the agents from the initial conditions $x_i(0)$, $x_j(0)$ to the desired states $x_{i, \text{des}}$, $x_{j, \text{des}}$, respectively. This forms an infeasible task since from these initial conditions there is no controller that can navigate the agents towards the desired states without colliding with the obstacles and without leaving the workspace.}
	\label{fig:infeasible_initial_conditions}
\end{figure}

\begin{definition} \label{def:set_feasible_initial_conditions}
	Let $x_{i, \text{des}} \in \mathbb{R}^n$, $i \in \mathcal{V}$ be a desired feasible configuration as given in Definition \ref{definition:feasible_steady_state_conf}, respectively. Then, the set of all initial conditions $x_i(0)$ according to Assumption \ref{ass:initial_conditions}, for which there exist time constants $\overline{t}_i \in \mathbb{R}_{> 0} \cup \{\infty\}$ and control inputs $u^\star_i \in \mathcal{U}_i$, $i \in \mathcal{V}$, which define a solution $x_i^\star(t)$, $t \in [0, \overline{t}_i]$ of the system \eqref{eq:system}, under the presence of disturbance $w_i \in \mathcal{W}_i$, such that:
	\begin{enumerate}
	\item $x_i^\star(\overline{t}_i) = x_{i, \text{des}}$;
	\item $\| x_i^\star(t) - x_j^\star(t) \| > r_{i} + r_{j}$, for every $t \in [0, \overline{t}_i]$, $i, j \in \mathcal{V}$, $i \neq j$;
	\item $\|x^\star_i(t) - x_{\scriptscriptstyle O_\ell}\| > r_{i} + r_{\scriptscriptstyle O_\ell}$, for every $t \in [0, \overline{t}_i]$, $i \in \mathcal{V}$, $\ell \in \mathcal{L}$;
	\item $\|x_{\scriptscriptstyle \mathcal{D}}-x^\star_i(t)\| < r_{\scriptscriptstyle \mathcal{D}} - r_i$, for every $t \in [0, \overline{t}_i]$, $i \in \mathcal{V}$;
	\item $\|x^\star_i(t) - x^\star_j(t)\| < d_i$, for every $t \in [0, \overline{t}_i]$, $i \in \mathcal{V}, j \in \mathcal{N}_i$,
	\end{enumerate}
	are called \emph{feasible initial conditions}.
\end{definition}
The feasible initial conditions are essentially all the initial conditions $x_i(0) \in \mathbb{R}^n$, $i \in \mathcal{V}$ from which there exist controllers $u_i \in \mathcal{U}_i$ that can navigate the agents to the given desired states $x_{i, \text{des}}$, under the presence of disturbances $w_i \in \mathcal{W}_i$ while the initial neighbors remain connected, the agents do not collide with each other, they stay in the workspace and they do not collide with the obstacles of the environment. Initial conditions for which one or more agents can not be driven to the desired state $x_{i, \text{des}}$ by a controller $u_i \in \mathcal{U}_i$, i.e., initial conditions that violate one or more of the conditions of Definition \ref{def:set_feasible_initial_conditions}, are considered as \emph{infeasible initial conditions}. An example with infeasible initial conditions is depicted in Fig. \ref{fig:infeasible_initial_conditions}. 

\subsection{Problem Statement}

\noindent Formally, the control problem, under the aforementioned constraints, is formulated as follows:
\begin{problem} \label{problem}
	Consider $N$ agents governed by dynamics as in \eqref{eq:system}, modeled by the balls $\mathcal{B}\left(x_i, r_i\right)$, $i \in \mathcal{V}$, operating in a workspace $\mathcal{D}$ which is modeled by the ball $\mathcal{B}\left(x_{\scriptscriptstyle \mathcal{D}}, r_{\scriptscriptstyle \mathcal{D}}\right)$. In the workspace there are $L$ obstacles $\mathcal{B}\left(x_{\scriptscriptstyle O_\ell}, r_{\scriptscriptstyle O_\ell}\right)$, $\ell \in \mathcal{L}$. The agents have communication capabilities according to Assumption \ref{ass:measurements_access}, under the initial conditions $x_i(0)$, imposed by Assumption \ref{ass:initial_conditions}. Then, given a desired feasible configuration $x_{i, \text{des}}$ according to Definition \ref{definition:feasible_steady_state_conf}, for all feasible initial conditions, as given in Definition \ref{def:set_feasible_initial_conditions}, the problem lies in designing \emph{decentralized feedback control} laws $u_i \in \mathcal{U}_i$, such that for every $i \in \mathcal{V}$ and for all times $t \in \mathbb{R}_{\geq 0}$, the following specifications are satisfied: 
	\begin{enumerate}
		\item Position stabilization is achieved: $\displaystyle \lim_{t \to \infty} \|x_i(t) - x_{i, \text{des}} \| \to 0;$
		\item Inter-agent collision avoidance: $\|x_i(t) - x_j(t)\| > r_{i} + r_{j}, \forall \ j \in \mathcal{V} \backslash \{i\};$
		\item Connectivity maintenance between neighboring agents is preserved: $\|x_i(t) - x_j(t)\| < d_i$, $\forall \ j \in \mathcal{N}_i;$
		\item Agent-with-obstacle collision avoidance: $\|x_i(t) - x_{\scriptscriptstyle O_\ell}(t)\| > r_{i} + r_{\scriptscriptstyle O_\ell}$, $\forall \ \ell \in \mathcal{L};$
		\item Agent-with-workspace-boundary collision avoidance: $ \|x_{\scriptscriptstyle \mathcal{D}}-x_i(t)\| < r_{\scriptscriptstyle \mathcal{D}} - r_i$.
	\end{enumerate}
\end{problem}

\section{Main Results} \label{sec:main_results}

In this section, a systematic solution to Problem \ref{problem} is introduced. Our overall approach builds on designing a decentralized control law $u_i \in \mathcal{U}_i$ for each agent $i \in \mathcal{V}$. In particular, since we aim to minimize the norms $\|x_i(t) - x_{i, \text{des}} \|$, as $t \to \infty$ subject to the state constraints imposed by Problem \ref{problem}, it is reasonable to seek a solution which is the outcome of an optimization problem.

\subsection{Error Dynamics and Constraints} \label{sec:error_dynamics}

Define the error vector $e_i: \mathbb{R}_{\ge 0} \to \mathbb{R}^n$ by: $e_i(t) = x_i(t)-x_{i, \text{des}}$. Then, the \emph{error dynamics} are given by:
\begin{equation} \label{eq:error_system_perturbed}
\dot{e}_i(t) = h_i(e_i(t), u_i(t)),
\end{equation}
where the functions $h_i: \mathbb{R}^n \times \mathbb{R}^m \to \mathbb{R}^{n}$, $g_i: \mathbb{R}^n \times \mathbb{R}^m \to \mathbb{R}^{n}$ are defined by:
\begin{subequations}
\begin{align}
h_i(e_i(t), u_i(t)) & \triangleq g_i(e_i(t), u_i(t)) + w_i(e_i(t)+x_{i, \text{des}}, t), \\
g_i(e_i(t), u_i(t)) & \triangleq f_i(e_i(t)+x_{i, \text{des}}, u_i(t)). \label{eq:function_g_i}
\end{align}
\end{subequations}
Define the set that captures all the \textit{state} constraints on the system \eqref{eq:system}, posed by Problem \ref{problem} by:
\begin{align*}
\mathcal{Z}_{i} \triangleq \Big\{ x_i(t) \in \mathbb{R}^n : \  & \|x_i(t) - x_j(t)\| \ge r_i+r_j+\varepsilon, \forall \ j \in \mathcal{R}_i(t), \notag \\ 
& \|x_i(t) - x_j(t)\| \le d_i - \varepsilon, \forall \ j \in \mathcal{N}_i, \notag \\ 
& \|x_i(t) - x_{\scriptscriptstyle O_\ell}\| \ge r_{i} + r_{\scriptscriptstyle O_\ell} + \varepsilon, \forall \ \ell \in \mathcal{L}, \notag \\ 
& \|x_{\scriptscriptstyle \mathcal{D}}-x_i(t)\| \le r_{\scriptscriptstyle \mathcal{D}} - r_i - \varepsilon \Big\}, i \in \mathcal{V}, 
\end{align*}
where $\varepsilon \in \mathbb{R}_{> 0}$ is an arbitrary small constant. In order to translate the constraints that are dictated for the state $z_i$ into constraints regarding the error state $e_i$, define the set $$\mathcal{E}_{i} = \left\{e_i \in \mathbb{R}^n :
e_i \in \mathcal{Z}_{i} \oplus (-x_{i, \text{des}}) \right\}, i \in \mathcal{V}.$$ Then, the following equivalence holds: $x_i \in \mathcal{Z}_i \Leftrightarrow e_i \in \mathcal{E}_i$, $\forall i \in \mathcal{V}$.

\begin{property} \label{property 1}
	The nonlinear functions $g_i$, $i \in \mathcal{V}$ are \emph{Lipschitz continuous} in $\mathcal{E}_i \times \mathcal{U}_i$ with Lipschitz constants $L_{g_i} = L_{f_i}$. Thus, it  holds that: $$\|g_i(e, u)-g_i(e', u)\| \le L_{g_i} \|e-e'\|, \forall e, e' \in \mathcal{E}_i, u \in \mathcal{U}_i.$$
	\begin{proof}
		The proof can be found in Appendix \ref{app:proof_property_2}.
	\end{proof}
\end{property}		

If the decentralized control laws $u_i \in \mathcal{U}_i$, $i \in \mathcal{V}$, are designed such that the error signal $e_i$ with dynamics given in \eqref{eq:error_system_perturbed}, constrained under $e_i \in \mathcal{E}_{i}$, satisfies $\lim\limits_{t \to \infty} \|e_i(t)\| \to 0$, then Problem \ref{problem} will have been solved. 

\subsection{Decentralized Control Design} \label{sec:control_design}

Due to the fact that we have to deal with the minimization of norms $\|e_i(t) \|$, as $t \to \infty$, subject to constraints $e_i \in \mathcal{E}_i$, we invoke here a class of Nonlinear Model Predictive controllers. NMPC frameworks have been studied in \cite{Mayne2000789, morrari_npmpc, frank_1998_quasi_infinite, cannon_2001_nmpc, camacho_2007_nmpc, fontes_2001_nmpc_stability, frank_2003_towards_sampled-data-nmpc, fontes_2007_modified_barbalat, borrelli_2013_nmpc, grune2016nonlinear, IJC_2017} and they have been proven to be powerful tools for dealing with state and input constraints.

Consider a sequence of sampling times $\{t_k\}$, $k \in \mathbb{N}$, with a constant sampling time $h$, $0 < h < T_p$, where $T_p$ is the prediction horizon, such that $t_{k+1} = t_k + h$, $\forall k \in \mathbb{N}$. Hereafter we will denote by $i$ the agent and by index $k$ the sampling instant. In sampled data NMPC, a Finite-Horizon Open-loop Optimal Control Problem (FHOCP) is solved at the discrete sampling time instants $t_k$ based on the current state error measurement $e_i(t_k)$. The solution is an optimal control signal $\overline{u}_i^{\star}(s)$, computed over $s \in [t_k, t_k+T_p]$. The open-loop input signal applied in between the sampling instants is given by the solution of the following FHOCP:
\begin{subequations}
	\begin{align}
	&\hspace{-1mm}\min\limits_{\overline{u}_i(\cdot)} J_i(e_i(t_k), \overline{u}_i(\cdot)) \notag \\ 
	&\hspace{-1mm}= \min\limits_{\overline{u}_i(\cdot)} \left\{  V_i(\overline{e}_i(t_k+T_p)) + \int_{t_k}^{t_k+T_p} \Big[ F_i(\overline{e}_i(s), \overline{u}_i(s)) \Big] ds \right\}  \label{mpc_position_based_cost_2}\\
	&\hspace{-1mm}\text{subject to:} \notag \\
	&\hspace{1mm} \dot{\overline{e}}_i(s) = g_i(\overline{e}_i(s), \overline{u}_i(s)), \overline{e}_i(t_k) = e_i(t_k), \label{eq:diff_mpc} \\
	&\hspace{1mm} \overline{e}_i (s) \in \mathcal{E}_{i, s - t_k}, \overline{u}_i(s) \in \mathcal{U}_i, s \in [t_k,t_k+T_p], \label{eq:mpc_constrained_set} \\
	&\hspace{1mm} \overline{e}(t_k+T_p)\in \Omega_i. \label{eq:mpc_terminal_set}
	\end{align}
\end{subequations}
At a generic time $t_k$ then, agent $i \in \mathcal{V}$ solves the aforementioned FHOCP. The notation $\overline{\cdot}$ is used to distinguish predicted states which are internal to the controller, corresponding to the nominal system \eqref{eq:diff_mpc} (i.e., the system \eqref{eq:error_system_perturbed} by substituting $w(\cdot) = 0_{n \times 1}$). This means that $\overline{e}_i(\cdot)$ is the solution to \eqref{eq:diff_mpc} driven by the control input $\overline{u}_i(\cdot) : [t_k, t_k + T_p] \to \mathcal{U}_i$ with initial condition $e_i(t_k)$. Note that the predicted states are not the same with the actual closed-loop values due to the fact that the system is under the presence of disturbances $w_i \in \mathcal{W}_i$. The functions $F_i : \mathcal{E}_{i} \times \mathcal{U}_i \to \mathbb{R}_{\geq 0}$, $V_i: \mathcal{E}_i \to \mathbb{R}_{\geq 0}$ stand for the \emph{running costs} and the \emph{terminal penalty costs}, respectively, and they are defined by: $F_i \big(e_i, u_i\big) = e_i^{\top} Q_i e_i + u_i^{\top} R_i u_i$, $V_i \big(e_i\big) = e_i^{\top} P_i e_i$; $R_i \in \mathbb{R}^{m \times m}$ and $Q_i$, $P_i \in \mathbb{R}^{n \times n}$ are symmetric and positive definite controller gain matrices to be appropriately tuned; $Q_i \in \mathbb{R}^{n \times n}$ is a symmetric and positive semi-definite controller gain matrix to be appropriately tuned. The sets $\mathcal{E}_{i, s - t_k}$, $\Omega_i$ will be explained later. For the running costs $F_i$ the following hold:
\begin{lemma} \label{lemma:F_i_bounded_K_class}
There exist functions $\alpha_1$, $\alpha_2 \in \mathcal{K}_{\infty}$ such that: $$\alpha_1\big(\|\eta_i\|\big) \leq F_i\big(e_i, u_i\big) \leq \alpha_2\big(\|\eta_i \|\big),$$ for every $\eta_i \triangleq \left[ e_i^\top, u_i^\top\right]^\top \in \mathcal{E}_{i} \times \mathcal{U}_i$, $i \in \mathcal{V}$.
\end{lemma}
\begin{proof}
	The proof can be found in Appendix \ref{app:proof_lemma_1}.
\end{proof}
\begin{lemma}\label{lemma:F_Lipschitz}
	The running costs $F_i$, $i \in \mathcal{V}$ are Lipschitz continuous in $\mathcal{E}_{i} \times \mathcal{U}_i$. Thus, it holds that: $$\big|F_i(e, u) - F_i(e', u)\big| \leq L_{F_i} \|e - e'\|, \forall e, e' \in \mathcal{E}_i, u \in \mathcal{U}_i,$$
	where $L_{F_i} \triangleq 2 \sigma_{\max}(Q_i) \sup\limits_{e \in \mathcal{E}_{i}} \|e\|$.
\end{lemma}
\begin{proof}
	The proof can be found in Appendix \ref{app:proof_of_F_lipsitz}.
\end{proof}

The applied input signal is a portion of the optimal solution to an
optimization problem where information on the states of the neighboring
agents of agent $i$ is taken into account only in the constraints considered
in the optimization problem. These constraints pertain to the set of its
neighbors $\mathcal{N}_i$ and, in total, to the set of all agents within its
sensing range $\mathcal{R}_i$. Regarding these, we make the following assumption:
\begin{assumption}
	\label{ass:access_to_predicted_info_n}
	When at time $t_k$ agent $i$ solves a FHOCP, it has access to the following measurements, across the entire horizon $s \in (t_k, t_k + T_p]$: 
	\begin{enumerate}
		\item Measurements of the states:
		\begin{itemize}
			\item $x_j(t_k)$ of all agents $j \in \mathcal{R}_i(t_k)$ within its sensing range at time $t_k$;
			\item $x_{j'}(t_k)$ of all of its neighboring agents $j' \in \mathcal{N}_i$ at time $t_k$;
		\end{itemize}
		\item The \textit{predicted states}:
		\begin{itemize}
			\item $\overline{x}_j(s)$ of all agents $j \in \mathcal{R}_i(t_k)$ within its sensing range;
			\item $\overline{x}_{j'}(s)$ of all of its neighboring agents $j' \in \mathcal{N}_i$;
		\end{itemize}
	\end{enumerate}
\end{assumption}

In other words, each time an agent solves its own individual optimization
problem, it knows the (open-loop) state predictions that have been generated
by the solution of the optimization problem of all agents within
its range at that time, for the next $T_p$ time units. This assumption is
crucial to satisfying the constraints regarding collision avoidance and
connectivity maintenance between neighboring agents.
We assume that the above pieces of information are \emph{always available},
\emph{accurate} and can be exchanged \emph{without delay}. 

\begin{figure}[t!]
	\centering
	\includegraphics[scale = 0.36]{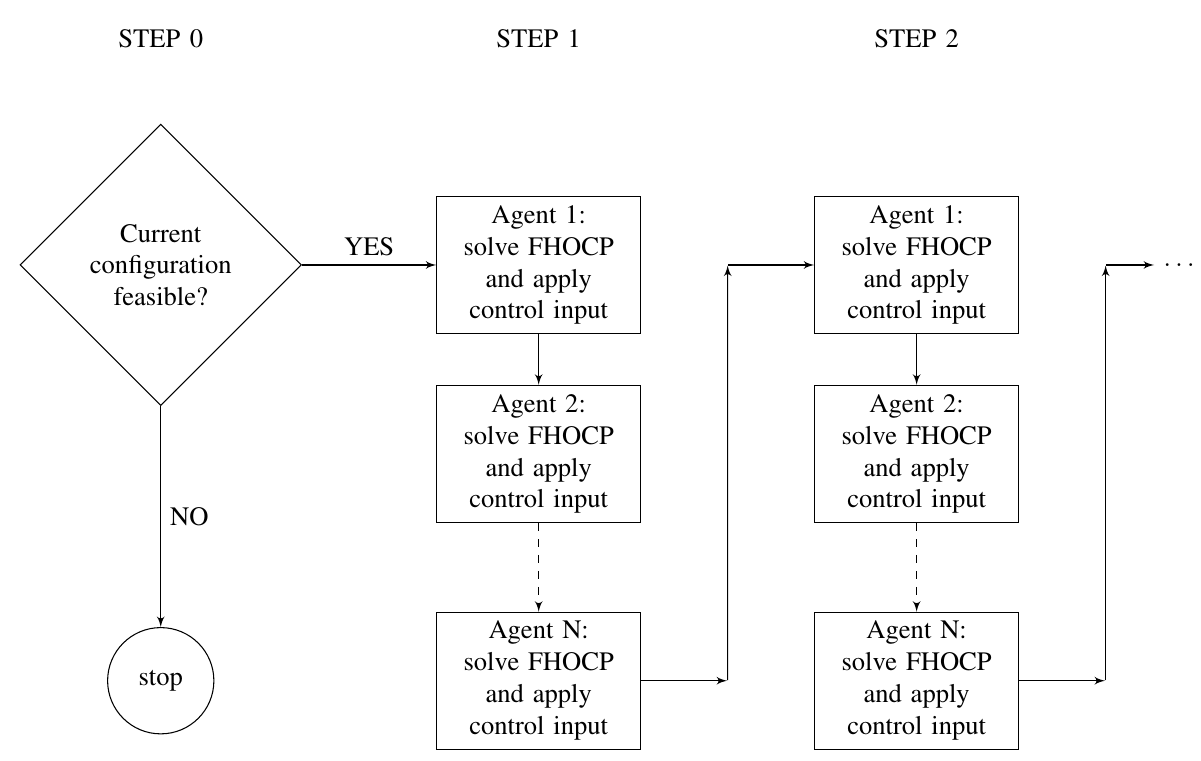}
	\caption{The procedure is approached sequentially. Notice that the
		figure implies that recursive feasibility is established if the initial
		configuration is itself feasible.}
	\label{fig:process_flow}
\end{figure}

\begin{remark}
	The designed procedure flow can be either concurrent or sequential,
	meaning that agents can solve their individual FHOCP’s and apply the control
	inputs either simultaneously, or one after the other. The conceptual
	design itself is procedure-flow agnostic, and hence it can incorporate both
	without loss of feasibility or successful stabilization. The approach that we
	have adopted here is the sequential one: each agent solves its own FHOCP and
	applies the corresponding admissible control input in a round robin way,
	considering the current and planned (open-loop state predictions)
	configurations of all agents within its sensing range.
	Figure \ref{fig:process_flow} and Figure \ref{fig:information_flow} depict
	the sequential procedural and informational regimes.
\end{remark}

The constraint sets $\mathcal{E}_i$, $i \in \mathcal{V}$ involve the sets $\mathcal{R}_i(t)$ which are updated at every sampling time in which agent $i$ solves his own optimization problem. Its predicted configuration at time $s \in [t_k, t_k + T_p]$ is constrained by the predicted configuration of its neighboring and perceivable agents (agents within its sensing range) at the same time instant $s$.  

The solution to FHOCP \eqref{mpc_position_based_cost_2} - \eqref{eq:mpc_terminal_set} at time $t_k$ provides an optimal control input, denoted by
$\overline{u}_i^{\star}(s;\ e_i(t_k))$, $s \in [t_k, t_k + T_p]$. This control input is then applied to the system until the next sampling instant $t_{k+1}$:
\begin{align}
u_i(s; e_i(t_k)) = \overline{u}_i^{\star}\big(s; \ e_i(t_k)\big),\  s \in [t_k, t_{k+1}).
\label{eq:position_based_optimal_u_2}
\end{align}
At time $t_{k+1}$ a new finite horizon optimal control problem is solved in the
same manner, leading to a receding horizon approach. The control input $u_i(\cdot)$ is of feedback form, since it is recalculated at each sampling instant based on the then-current state. The solution of \eqref{eq:error_system_perturbed} at time $s$, $s \in [t_k, t_k+T_p]$, starting at time $t_k$, from an initial condition $e_i(t_k) = \overline{e}_i(t_k)$, by application of the control input $u_i : [t_k, s] \to \mathcal{U}_i$ is denoted by $e_i\big(s;\ u_i(\cdot), e_i(t_k)\big)$, $s \in [t_k, t_k+T_p]$. The \textit{predicted} state of the system \eqref{eq:diff_mpc} at time $s$, $s \in [t_k, t_k+T_p]$ based on the measurement of the state at time $t_k$, $e_i(t_k)$, by application of the control input $u_i\big(t;\ e_i(t_k)\big)$ as in \ref{eq:position_based_optimal_u_2}, is denoted by $\overline{e}_i\big(s;\ u_i(\cdot), e_i(t_k)\big)$, $s \in [t_k, t_k+T_p]$.  Due to the fact that the system is in presence of disturbances $w_i \in \mathcal{W}_i$, it holds that: $\overline{e}_i(\cdot) \neq e_i(\cdot)$.
\begin{property}
	By integrating \eqref{eq:error_system_perturbed}, \eqref{eq:diff_mpc} at the
	time interval $s \ge \tau$, the actual $e_i(\cdot)$ and the predicted
	states $\overline{e}_i(\cdot)$ are respectively given by:
	\begin{subequations}
		\begin{align}
		e_i\big(s;\ u_i(\cdot), e_i(\tau)\big) &=
		e_i(\tau) + \int_{\tau}^{s} h_i\big(e_i(s';\ e_i(\tau)), u_i(s)\big) ds', \label{eq:remark_4_eq_1} \\
		\overline{e}_i\big(s;\ u_i(\cdot), e_i(\tau)\big) &=
		e_i(\tau) + \int_{\tau}^{s} g_i\big(\overline{e}_i(s';\ e_i(\tau)), u_i(s')\big) ds'. \label{eq:remark_4_eq_2}
		\end{align}
	\end{subequations}
	\label{remark:predicted_actual_equations_with_disturbance}
\end{property}
\begin{figure}[t!]\centering
	\scalebox{0.6}{
\tikzstyle{decision} = [diamond, draw, 
    text width=7.5em, text badly centered, node distance=3cm, inner sep=0pt]
\tikzstyle{block} = [rectangle, draw, 
    text width=10em, text centered, minimum height=4em]
\tikzstyle{block_rounded} = [rectangle, rounded corners, draw, 
    text width=12em, text centered, minimum height=4em]
\tikzstyle{line} = [draw, -latex']
\tikzstyle{cloud} = [draw, ellipse, 
    node distance=3cm, minimum height=4em, minimum width=4em]

\begin{tikzpicture}[node distance = 2cm, auto]


    \node [block, node distance=1.5cm] (system_m) {Agent $m \in \mathcal{R}_i(t_k)$};
    \node [block, below of=system_m, node distance=2.5cm] (system_n) {Agent $n \in \mathcal{R}_i(t_k)$};

    \node[right of=system_n, node distance=4.5cm] (right_of_n){};

    \node [block_rounded, below of=right_of_n, node distance=2.5cm] (latest_plans) {Latest predictions (current timestep $t_k$)};
    \node [block, below of=latest_plans, node distance=2.5cm] (agent_i) {Agent $i$};
    \node [block_rounded, below of=agent_i, node distance=2.5cm] (latest_plans_) {Latest predictions \\ (previous timestep $t_{k-1}$)};

    \node[left of=latest_plans_, node distance=4.5cm] (left_of_i){};

    \node [block, below of=left_of_i, node distance=2.5cm] (system_p) {Agent $p \in \mathcal{R}_i(t_k)$};
    \node [block, below of=system_p, node distance=2.5cm] (system_q) {Agent $q \in \mathcal{R}_i(t_k)$};

    \node [above of=system_m, node distance=1.5cm] (dots_0) {$\vdots$};
    \node [below of=system_q, node distance=1.5cm] (dots_1) {$\vdots$};

    \node[inner sep=0,minimum size=0,right of=system_m, node distance=5.5cm] (sm) {};
    \node[inner sep=0,minimum size=0,right of=system_n, node distance=4.5cm] (sn) {};
    \node[inner sep=0,minimum size=0,right of=system_p, node distance=4.5cm] (sp) {};
    \node[inner sep=0,minimum size=0,right of=system_q, node distance=5.5cm] (sq) {};

    \path[line]           (system_m) -- (sm);
    \path[line]           (system_n) -- (sn);
    \path[line]           (sm) -- (latest_plans.38);
    \path[line]           (sn) -- (latest_plans);

    \path[line]           (system_p) -- (sp);
    \path[line]           (system_q) -- (sq);
    \path[line]           (sp) -- (latest_plans_);
    \path[line]           (sq) -- (latest_plans_.-38);

    \path[line,dashed]    (latest_plans) -- (agent_i);
    \path[line,dashed]    (latest_plans.-38) -- (agent_i.38);

    \path[line,dashed]    (latest_plans_) -- (agent_i);
    \path[line,dashed]    (latest_plans_.38) -- (agent_i.-38);




\end{tikzpicture}}
	\caption{The flow of information to agent $i$ regarding its perception of
		agents within its sensing range $\mathcal{R}_i$ at arbitrary FHOCP
		solution time $t_k$. Agents $m,n \in \mathcal{R}_i(t_k)$ have solved their
		FHOCP; agent $i$ is next; agents $p,q \in \mathcal{R}_i(t_k)$ have not
		solved their FHOCP yet.}
	\label{fig:information_flow}
\end{figure}
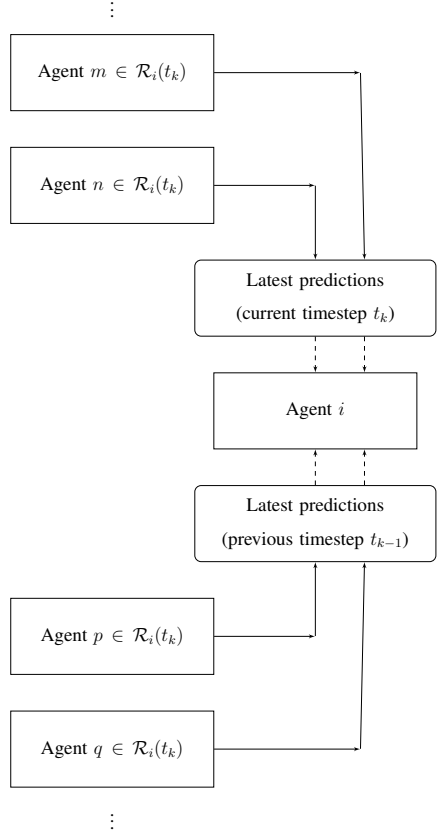

\vspace{-7mm}

The satisfaction of the constraints $\mathcal{E}_i$ on the state along the prediction horizon depends on the future realization of the uncertainties. On the assumption of additive uncertainty and Lipschitz continuity of the nominal model, it is possible to compute a bound on the future effect of the uncertainty on the system. Then, by considering this effect on the state constraint on the nominal prediction, it is possible to guarantee that the evolution of the real state of the system will be admissible all the time. In view of latter, the state constraint set $\mathcal{E}_i$ of the standard NMPC formulation, is being replaced by a restricted constrained set $\mathcal{E}_{s-t_k} \subseteq \mathcal{E}_i$ in \eqref{eq:mpc_constrained_set}. This state constraint's tightening for the nominal system \eqref{eq:diff_mpc} with additive disturbance $w_i \in \mathcal{W}_i$, is a key ingredient of the proposed controller and guarantees that the evolution of the evolution of the real system will be admissible for all times. If the state constraint set was left unchanged during the solution of the optimization problem, the applied input to the plant, coupled with the uncertainty affecting the states of the plant could force the states of the plant to escape their intended bounds. The aforementioned tightening set strategy is inspired by the works \cite{1185106, Fontes2007, alina_ecc_2011}, which have considered such a robust NMPC formulation.

\begin{lemma} \label{lemma:diff_state_from_same_conditions}
	The difference between the actual measurement $e_i\big(t_k + s;\ u_i(\cdot), e_i(t_k)\big)$ at time $t_k+s$, $s \in (0, T_p]$, and the predicted state $\overline{e}_i\big(t_k + s;\ u_i(\cdot), e_i(t_k)\big)$ at the same time, under a control input $u_i(\cdot) \in \mathcal{U}_i$, starting at the same initial state $e_i(t_k)$ is upper bounded by: $$\left\| e_i\big(t_k + s;\ u_i(\cdot), e_i(t_k)\big) -
	\overline{e}_i\big(t_k + s;\ u_i(\cdot), e_i(t_k)\big) \right\| \leq \dfrac{\widetilde{w}_i}{L_{g_i}} (e^{L_{g_i} s} - 1), s \in (0, T_p],$$ where $e^\cdot$ denotes the exponential function.
\end{lemma}
\begin{proof}
	The proof can be found in Appendix \ref{app:proof_lemma_diff_state}.
\end{proof}
\noindent By taking into consideration the aforementioned Lemma, the restricted constraints set are then defined by: $\mathcal{E}_{i, s-t_k} = \mathcal{E}_i \ominus \mathcal{X}_{i,s-t_k}$, with $$\mathcal{X}_{i,s-t_k} = \Big\{ e_i \in \mathbb{R}^n : \|e_i(s)\| \leq \dfrac{\widetilde{w}_i}{L_{g_i}}\big( e^{L_{g_i}(s - t_k)} - 1\big), \forall s \in [t_k, t_k + T_p] \Big\}.$$ This modification guarantees that the state of the real system $e_i$ is always satisfying the corresponding constraints $\mathcal{E}_i$. 
	\begin{property}
		\label{property:restricted_constraint_set}
		
		For every $s \in [t_k, t_k + T_p]$, it holds that if:
		\begin{align}
		\overline{e}_i\big( s;\ u_i(\cdot,\ e_i(t_k)), e_i(t_k) \big) \in \mathcal{E}_i \ominus \mathcal{X}_{i,s-t_k},
		\end{align}
		then the real state $e_i$ satisfies the constraints $\mathcal{E}_i$, i.e., $e_i(s) \in \mathcal{E}_i$.
	\end{property}
	\begin{proof}
		The proof can be found in Appendix \ref{app:proof_of_property_restricted_const}.
	\end{proof}
	
\begin{figure}[t!]
	\centering
	\begin{tikzpicture}[scale = 0.6, rotate=-30]
  \draw (2,2) ellipse (6cm and 3cm);
    \node at ($(2.5,2.5)+(75:6 and 3)$) {$\mathcal{E}_i$};
  \draw[dashed] (2,2) ellipse (5cm and 2.5cm);
    \node at ($(-2.3,1.2)+(75:5 and 2.5)$) {$\mathcal{E}_i \ominus \mathcal{X}_{i,t_{k+1} - t_k}$};
  \draw[dashed] (2,2) ellipse (2cm and 1cm);
    \node at ($(1.2,1.2)+(75:2 and 1)$) {$\mathcal{E}_i \ominus \mathcal{X}_{i, T_p}$};

  \node at (5.5,2) {$\dots$};
  \node at (-1.5,2) {$\dots$};

\pgflowlevel{\pgftransformrotate{30}}
\end{tikzpicture}
	\caption{The nominal constraint set $\mathcal{E}_i$ in bold and the
		consecutive restricted constraint sets $\mathcal{E}_i \ominus \mathcal{X}_{i, s-t_k}$,
		$s \in [t_k, t_k + T_p]$, dashed.}
	\label{fig:tightening_high_level}
\end{figure}
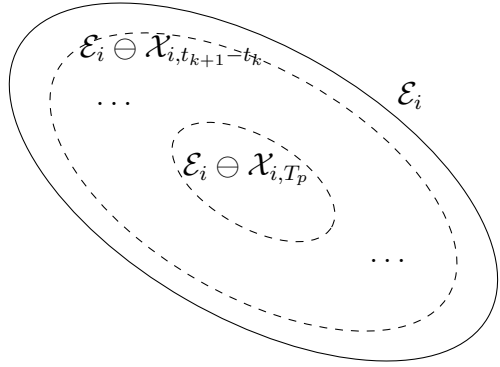

\begin{assumption}
	\label{ass:psi}
	The terminal set $\Omega_i \subseteq \Psi_i$ is a subset of an admissible and positively invariant set $\Psi_i$ as per Definition \ref{def:positively_invariant}, where $\Psi_i$ is defined as $\Psi_i \triangleq \big\{e_i \in \mathcal{E}_i : V_i(e_i) \leq \varepsilon_{\Psi_i} \big\},\ \varepsilon_{\Psi_i} > 0$.
\end{assumption}

\begin{assumption} \label{ass:psi_psi}
	The set $\Psi_i$ is interior to the set $\Phi_i$, $\Psi_i \subseteq \Phi_i$,
	which is the set of states within $\mathcal{E}_{i,T_p-h}$ for which
	there exists an admissible control input which is of linear feedback form with respect to the state $\kappa_i : [0,h] \to \mathcal{U}_i$: $\Phi_i \triangleq \big\{e_i \in \mathcal{E}_{i,T_p-h} : \kappa_i(e_i) \in \mathcal{U}_i \big\}$, such that for all $e_i \in \Psi_i$ and for all $s \in [0,h]$ it holds that: 
	\begin{equation} \label{eq:phi_psi}
	\dfrac{\partial V_i}{\partial e_i} g_i(e_i(s), \kappa_i(s))+ F_i(e_i(s), \kappa_i(s)) \leq 0.
	\end{equation}
\end{assumption}

\begin{remark} \label{remark:aux_control_stabilizability}
	According to \cite{262032, FINDEISEN2003190}, the existence of the linear state-feedback control law $\kappa_i$ is ensured if for every $i \in \mathcal{V}$ the following conditions hold:
	\begin{enumerate}
	\item $f_i$ is twice continuously differentiable with $f_i(0_{n \times 1}, 0_{m \times 1}) = 0_{n \times 1}$;
	\item Assumption \ref{ass:g_i_g_R_Lipschitz} holds;
	\item the sets $\mathcal{U}_i$ are compact with $0_{m \times 1} \in \mathcal{U}_i$, and
	\item the linearization of system \eqref{eq:error_system_perturbed} is stabilizable.
	\end{enumerate}
\end{remark}

\begin{assumption}
	\label{ass:psi_omega}
	The admissible and positively invariant set $\Psi_i$ is such that $\forall e_i(t) \in \Psi_i \Rightarrow e_i\big(t+s;\ \kappa_i(e_i(t)), e_i(t)\big) \in \Omega_i \subseteq \Psi_i$, for some $s \in [0,h]$.
\end{assumption}

\noindent The terminal sets $\Omega_i$ are chosen as: $\Omega_i \triangleq \big\{e_i \in \mathcal{E}_i : V_i(e_i)
\leq \varepsilon_{\Omega_i}\big\} \text{, where } \varepsilon_{\Omega_i} \in (0, \varepsilon_{\Psi_i})$.

\begin{figure}[t!]
	\centering
	\begin{tikzpicture}[scale = 0.7, rotate=-30]
  \draw[dashed](2,2) ellipse (5cm and 2.5cm);
    \node at ($(2.7,2.7)+(75:5 and 2.5)$) {$\mathcal{E}_i \ominus \mathcal{X}_{i,T_p-h}$};
  \draw[dashdotted](2,2) ellipse (4cm and 2cm);
    \node at ($(2.2,2.2)+(75:4 and 2)$) {$\Phi_i$};
  \draw[dashdotted] (2,2) ellipse (3cm and 1.5cm);
    \node at ($(2.2,2.2)+(75:3 and 1.5)$) {$\Psi_i$};
  \draw (2,2) ellipse (2cm and 1cm);
    \node at ($(2.2,2.2)+(75:2 and 1)$) {$\Omega_i$};
\pgflowlevel{\pgftransformrotate{30}}
\end{tikzpicture}
	\caption{The hierarchy of sets
		$\Omega_i \subseteq \Psi_i \subseteq \Phi_i \subseteq \mathcal{E}_{i,T_p-h}$,
		in bold, dash-dotted, dash-dotted, and dashed, respectively.
		For every state in $\Phi_i$ there is a linear state feedback control
		$\kappa_i(e_i)$ which, when applied to a state
		$e_i \in \Psi_i$, forces the trajectory of the state of the system to
		reach the terminal set $\Omega_i$.}
	\label{fig:tightening_low_level}
\end{figure}
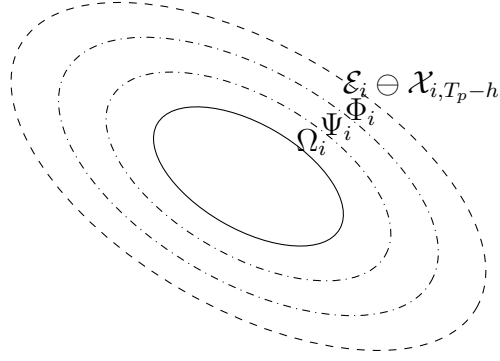

\begin{lemma} \label{lemma:V_i_lower_upper_bounded}
	For every $e_i \in \Psi_i$ there exist functions $\alpha_1, \alpha_2 \in \mathcal{K}_{\infty}$ such that: $$\alpha_1\big(\|e_i\|\big) \leq V_i(e_i) \leq \alpha_2\big(\| e_i \|\big), \forall i \in \mathcal{V}.$$
\end{lemma}
\begin{proof}
	The proof can be found in Appendix \ref{app:proof_lemma_V_i_lower_boudned}.
\end{proof}
\begin{lemma} \label{lemma:V_Lipschitz_e_0}
	The terminal penalty functions $V_i$ are Lipschitz continuous in $\Psi_i$, thus it holds that:
	$$\big|V_i(e) - V_i(e')\big| \leq L_{V_i} \|e - e'\|, \forall e, e' \in \Psi_i,$$
	where $L_{V_i} = 2 \sigma_{\max}(P_i)  \sup \limits_{e \in \Psi_i} \|e\|$.
\end{lemma}
\begin{proof}
	The proof is similar to the proof of Lemma \ref{lemma:F_Lipschitz} and it is omitted.
\end{proof}

\noindent We can now give the definition of an \textit{admissible input} for the FHOCP \eqref{mpc_position_based_cost_2}-\eqref{eq:mpc_terminal_set}.
\begin{definition}  \label{definition:admissible_input_with_disturbance}
	A control input $u_i : [t_k, t_k + T_p] \to \mathbb{R}^m$ for a state
	$e_i(t_k)$ is called \textit{admissible} for the FHOCP \eqref{mpc_position_based_cost_2}-\eqref{eq:mpc_terminal_set} if the following hold:
	\begin{enumerate}
		\item $u_i(\cdot)$ is piecewise continuous;
		\item $u_i(s) \in \mathcal{U}_i,\ \forall s \in [t_k, t_k + T_p]$;
		\item $e_i\big(t_k + s;\ u_i(\cdot), e_i(t_k)\big) \in \mathcal{E}_i \ominus \mathcal{X}_{i,s},\ \forall s \in [0, T_p]$ and
		\item $e_i\big(t_k + T_p;\ u_i(\cdot), e_i(t_k)\big) \in \Omega_i$.
	\end{enumerate}
\end{definition}

Under these considerations, we can now state the theorem that relates to the guaranteeing of the stability of the compound system of agents
$i \in \mathcal{V}$, when each of them is assigned a desired position.

\begin{theorem} \label{theorem}
	\label{theorem:with_disturbances}
	Suppose that for every $i \in \mathcal{V}$:
	\begin{enumerate}
		\item Assumptions \ref{ass:g_i_g_R_Lipschitz}-\ref{ass:psi_omega} hold;
		\item A solution to FHOCP \eqref{mpc_position_based_cost_2}-\eqref{eq:mpc_terminal_set} is feasible at time $t=0$ with feasible initial conditions, as defined in Definition \ref{def:set_feasible_initial_conditions};
		\item The upper bound $\widetilde{w}_i$ of the disturbance $w_i$ satisfies the following: $$\widetilde{w}_i \leq \dfrac{\varepsilon_{\Psi_i} - \varepsilon_{\Omega_i}}{\dfrac{L_{V_i}}{L_{g_i}} (e^{L_{g_i}h} - 1) e^{L_{g_i} (T_p - h)}}.$$
	\end{enumerate}
	Then, the closed loop trajectories of the system \eqref{eq:error_system_perturbed}, under the control input \eqref{eq:position_based_optimal_u_2} which is the outcome of the FHOCP \eqref{mpc_position_based_cost_2}-\eqref{eq:mpc_terminal_set}, converge to the set $\Omega_i$, as $t \to \infty$ and are ultimately bounded there, for every $i \in \mathcal{V}$.
\end{theorem}

\begin{proof}
	The proof of the theorem consists of two parts: firstly, recursive feasibility is established, that is, initial
	feasibility is shown to imply subsequent feasibility; secondly, and based
	on the first part, it is shown that the error state $e_i(t)$ reaches
	the terminal set $\Omega_i$ and it remains there for all times. The feasibility analysis and the convergence analysis can be found in Appendix \ref{app:feasibility_analysis} and Appendix \ref{app:convergence_analysis}, respectively.
\end{proof}

\begin{figure}[t!]
	\centering
	\includegraphics[scale = 0.50]{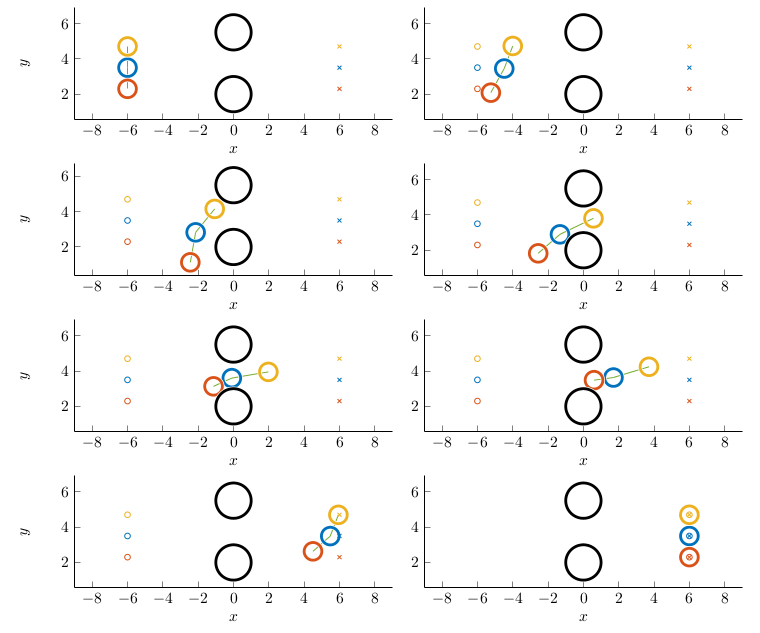}
	\caption{The trajectories of the three agents in the $x-y$ plane. Agent 1 is with
		blue, agent 2 with red and agent 3 with yellow. A faint green line connects
		agents deemed neighbors. The obstacles are depicted with black circles. The indicator ``O" denotes the configurations. The indicator ``X" marks desired configurations.}
	\label{fig:d_OFF_res_trajectory_3_2}
\end{figure}

\begin{remark}
	Due to the existence of disturbances, the error
	of each agent cannot be made to become arbitrarily close to zero, and
	therefore $\lim\limits_{t \to \infty} \|e_i(t)\|$ cannot converge to zero.
	However, if the conditions of Theorem 2 hold, then this error can be bounded
	above by the quantity $\sqrt{\varepsilon_{\Omega_i} / \lambda_{\max}(P_i)}$
	(since the trajectory of the error is trapped in the terminal set, this means
	that $V(e_i) = e_i^{\top} P_i e_i \leq \varepsilon_{\Omega_i}$ for every $e_i \in \Omega_i$).
\end{remark}

\section{Simulation Results} \label{sec:simulation_results}
For a simulation scenario, consider $N = 3$ unicycle agents with dynamics: $$\dot{x}_i(t) = 
\begin{bmatrix}
\dot{x}_i(t) \\
\dot{y}_i(t) \\ 
\dot{\theta}_i(t) \\
\end{bmatrix}
=
\begin{bmatrix}
v_i(t) \cos \theta_i(t) \\
v_i(t) \sin \theta_i(t) \\
\omega_i(t) \\
\end{bmatrix} 
+
w_i(x_i, t) I_{3 \times 1},$$
where: $i \in \mathcal{V} = \{1,2,3\}$, $x_i = \left[x_i, y_i, \theta_i \right]^\top$, $f_i(z_i, u_i) = \left[ v_i \cos \theta_i, v_i \sin \theta_i, \omega_i \right]^\top$, $u_i = \left[v_i, \omega_i \right]^\top$, $w_i = \widetilde{w}_i \sin(2 t)$, with $\widetilde{w}_i = 0.1$. We set $\widetilde{u}_i = 15$, $r_i = 0.5$, $d_i = 4r_i = 2.0$ and $\varepsilon = 0.01$. The neighboring sets are set to $\mathcal{N}_1 = \{2,3\}$, $\mathcal{N}_2 = \mathcal{N}_3 = \{1\}$. The agents' initial positions are $x_1$ $=$ $[-6, 3.5, 0]^{\top}$, $x_2$ $=$ $[-6, 2.3, 0]^{\top}$ and $x_3$ $=$ $[-6, 4.7, 0]^{\top}$. Their desired configurations in steady-state are $x_{1, \text{des}}$ $=$ $[6, 3.5, 0]^{\top}$, $x_{2, \text{des}}$ $=$ $[6, 2.3, 0]^{\top}$ and $x_{3, \text{des}}$ $=$ $[6, 4.7, 0]^{\top}$. In the workspace, we place $2$ obstacles with centers at points $[0, 2.0]^{\top}$ and $[0, 5.5]^{\top}$, respectively. The obstacles' radii are $r_{\scriptscriptstyle O_\ell} = 1.0$, $\ell \in \mathcal{L} = \{1,2\}$. The matrices $Q_i$, $R_i$, $P_i$ are set to $Q_i = 0.7 (I_3 + 0.5\dagger_3)$, $R_i = 0.005 I_2$ and $P_i = 0.5 (I_3 + 0.5\dagger_3)$, where $\dagger_N$ is a $N \times N$ matrix whose elements are uniformly randomly chosen between the values $0.0$ and $1.0$. The sampling time is $h = 0.1$ sec, the time-horizon is $T_p = 0.6$ sec, and the total execution time given is $10$ sec. Furthermore, we set: $L_{f_i} = 10.7354$, $L_{V_i} = 0.0471$, $\varepsilon_{\Psi_i} = 0.0654$ and $\varepsilon_{\Omega_i} = 0.0035$ for all $i \in \mathcal{V}$.
\begin{figure}[t!]
	\centering
	\includegraphics[scale = 0.40]{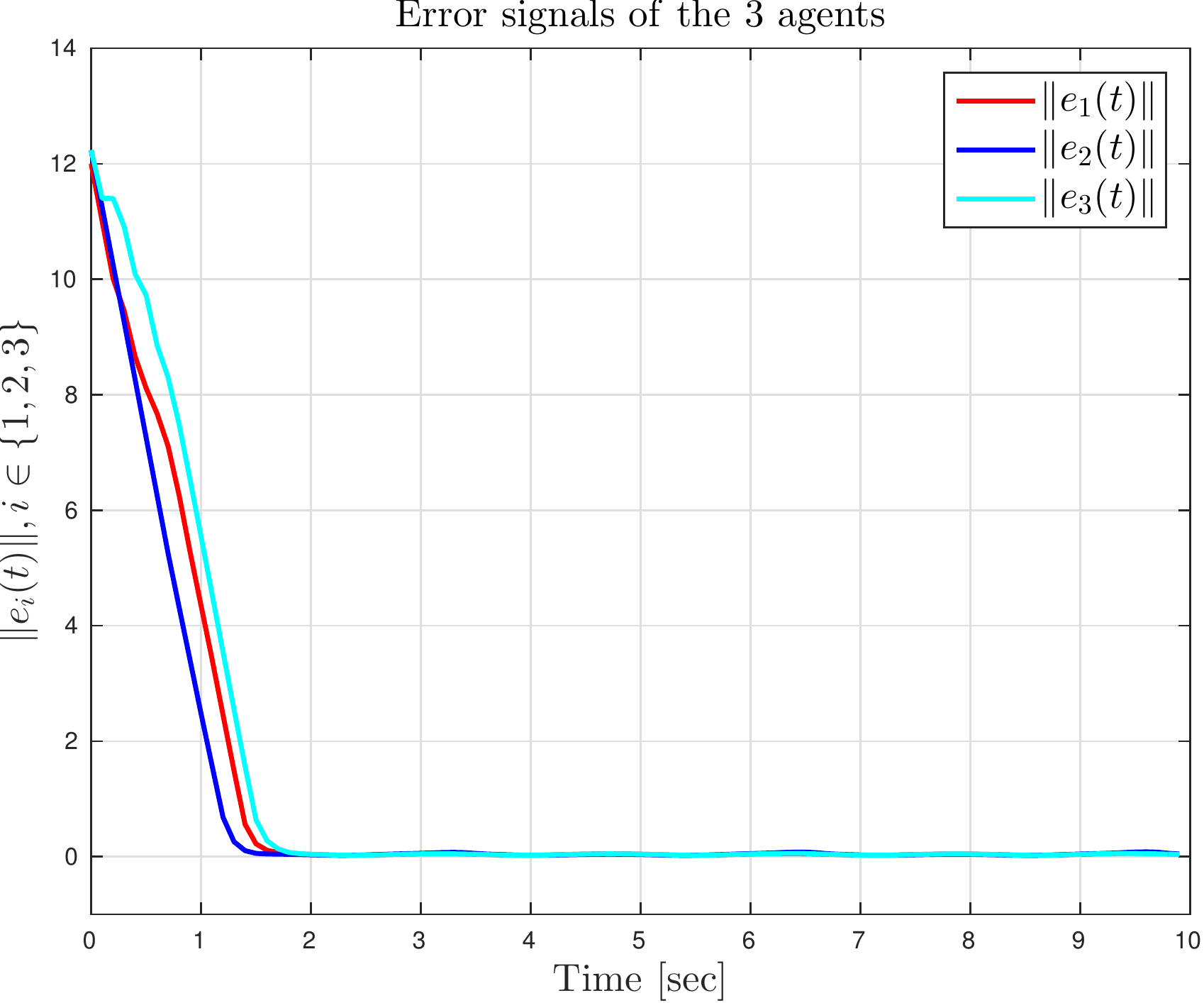}
	\caption{The evolution of the error signals of the three agents.}
	\label{fig:d_ON_res_3_2_errors_agent_1}
\end{figure}
\begin{figure}[t!]
	\centering
	\includegraphics[scale = 0.40]{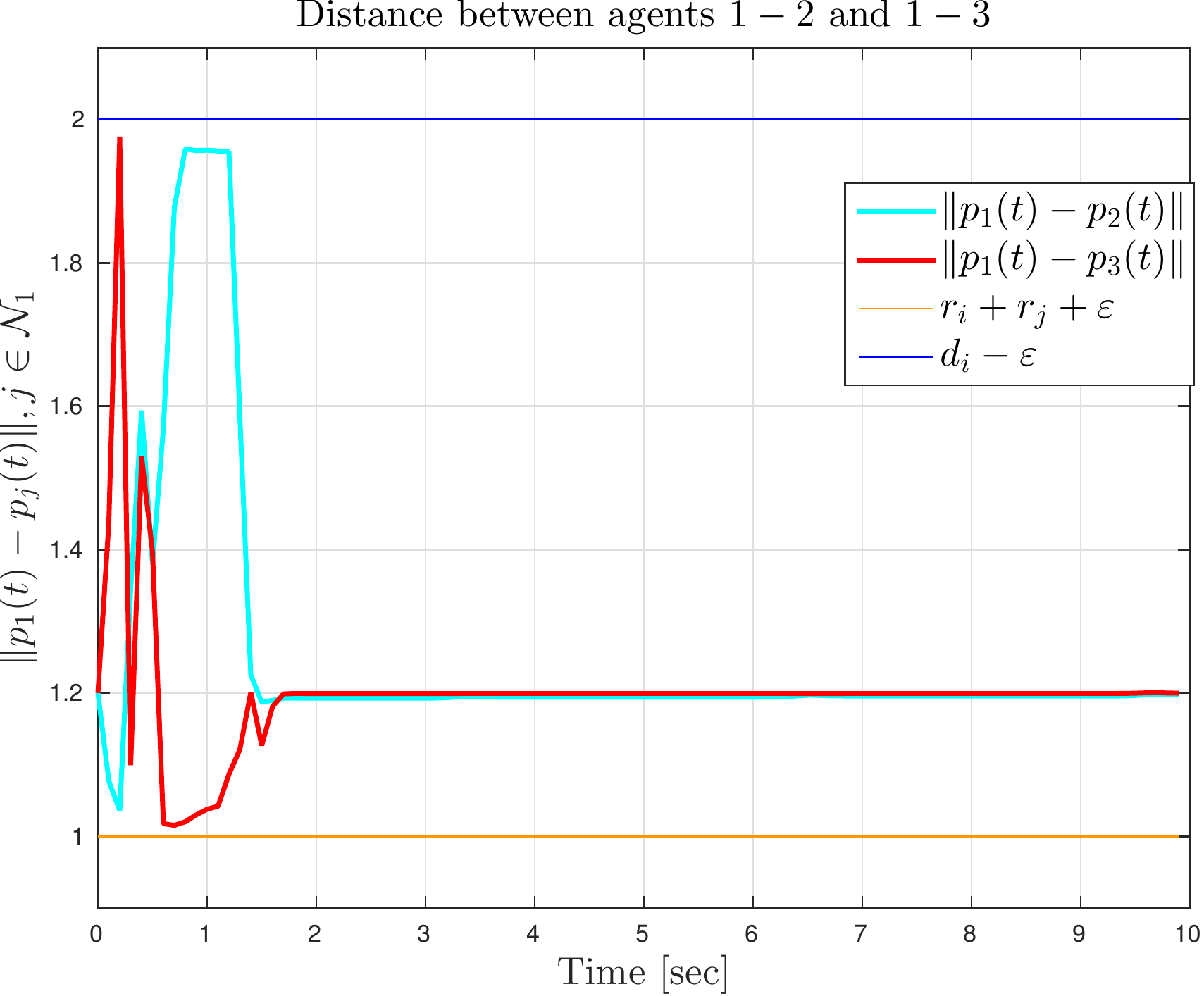}
	\caption{The distance between the agents $1-2$ and $1-3$ over time. The maximum and the minimum allowed
		distances are $d_i-\varepsilon = 1.99$ and $r_i+r_j +\varepsilon = 1.01$, respectively for every $i \in \mathcal{V}$, $j \in \mathcal{N}_i$.}
	\label{fig:d_ON_res_3_2_distance_agents_13}
\end{figure}
\begin{figure}[t!]
	\centering
	\includegraphics[scale = 0.40]{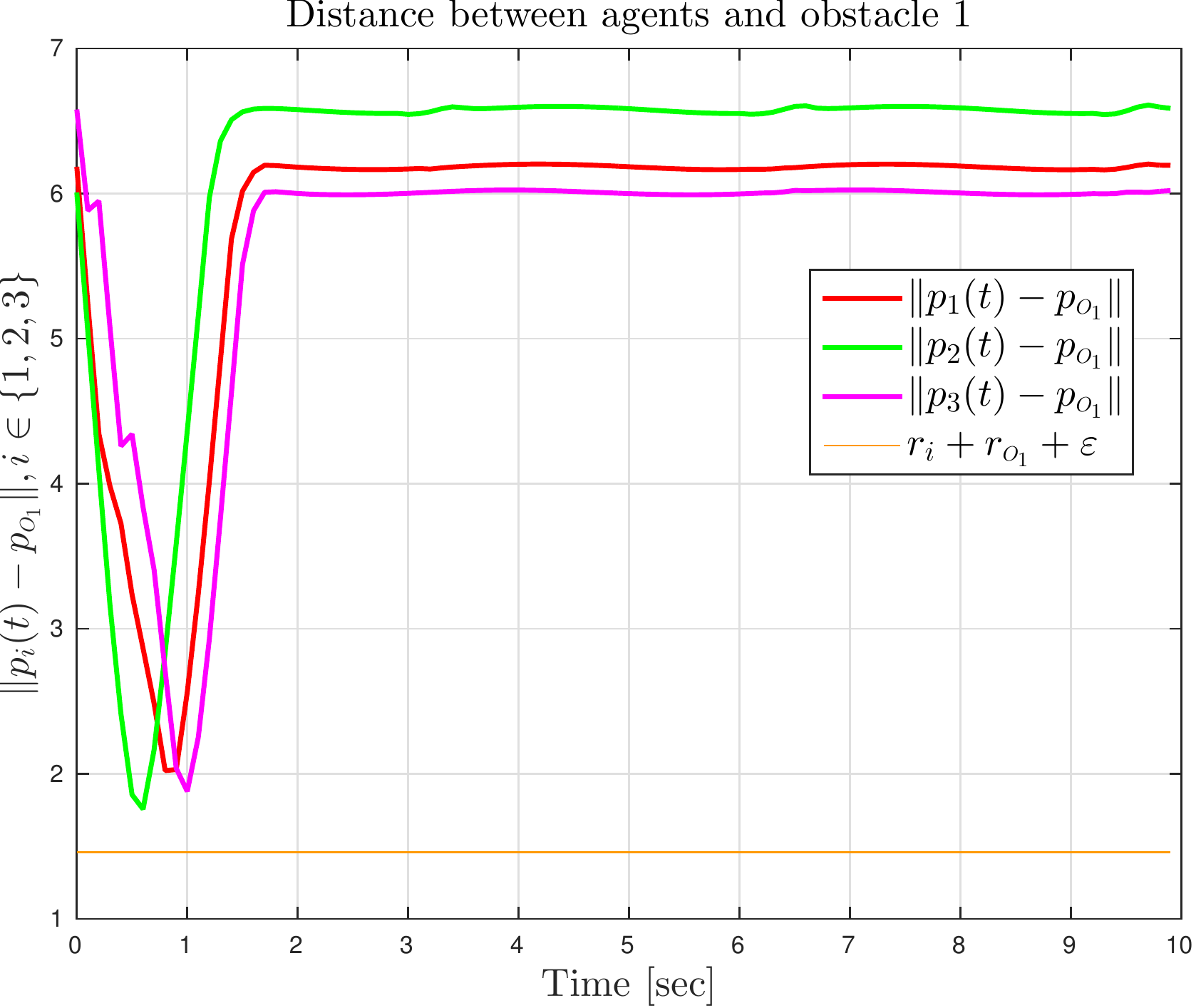}
	\caption{The distance between the agents and obstacle $1$ over time. The
		minimum allowed distance is $r_i + r_{\scriptscriptstyle O_1} + \varepsilon = 1.51$.}
	\label{fig:d_ON_res_3_2_distance_obstacle_1_agents}
\end{figure}

The frames of the evolution of the trajectories of the three agents in the $x-y$
plane are depicted in Figure \ref{fig:d_OFF_res_trajectory_3_2};  Figure \ref{fig:d_ON_res_3_2_errors_agent_1} depicts the evolution of the
error states of agents;  Figure \ref{fig:d_ON_res_3_2_distance_agents_13} shows the evolution of the distances between the neighboring agents; Figure \ref{fig:d_ON_res_3_2_distance_obstacle_1_agents} and Figure \ref{fig:d_ON_res_3_2_distance_obstacle_2_agents} depict the distance between the agents and the obstacle $1$ and $2$, respectively. Finally,  Figure \ref{fig:d_ON_res_3_2_inputs_agent_2} shows the input signals directing the agents through time. It can be observed that all agents reach their desired goal by satisfying all the constraints imposed by Problem \ref{problem}. The simulation was performed in MATLAB R2015a Environment utilizing the NMPC optimization routine provided in \cite{grune2016nonlinear}. The simulation takes $1340 \sec$ on a desktop with 8 cores, 3.60GHz CPU and 16GB of RAM.

\begin{figure}[t!]
	\centering
	\includegraphics[scale = 0.40]{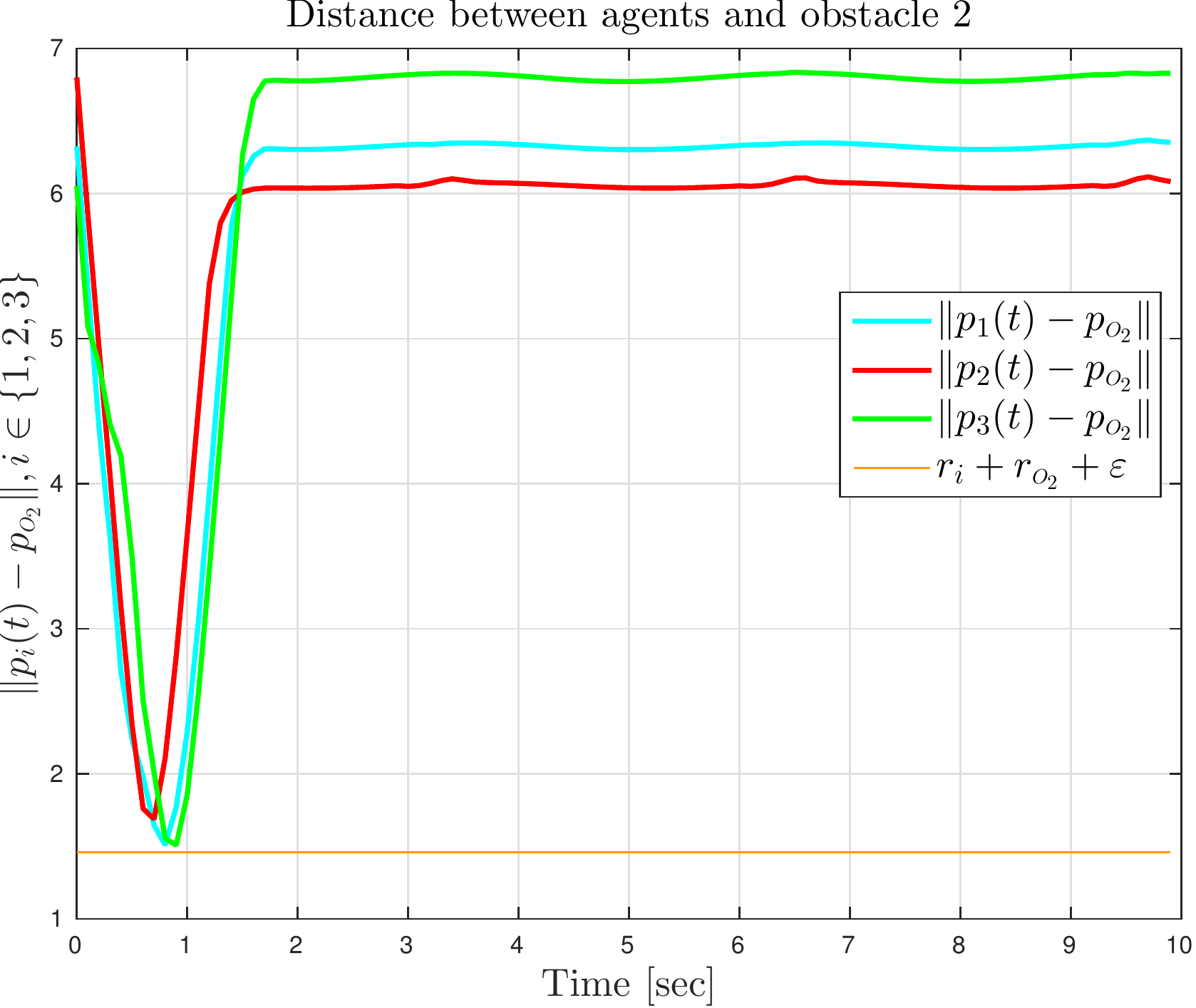}
	\caption{The distance between the agents and obstacle $2$ over time. The
		minimum allowed distance is $r_i + r_{\scriptscriptstyle O_2} + \varepsilon = 1.51$.}
	\label{fig:d_ON_res_3_2_distance_obstacle_2_agents}
\end{figure}

\begin{figure}[t!]
	\centering
	\includegraphics[scale = 0.40]{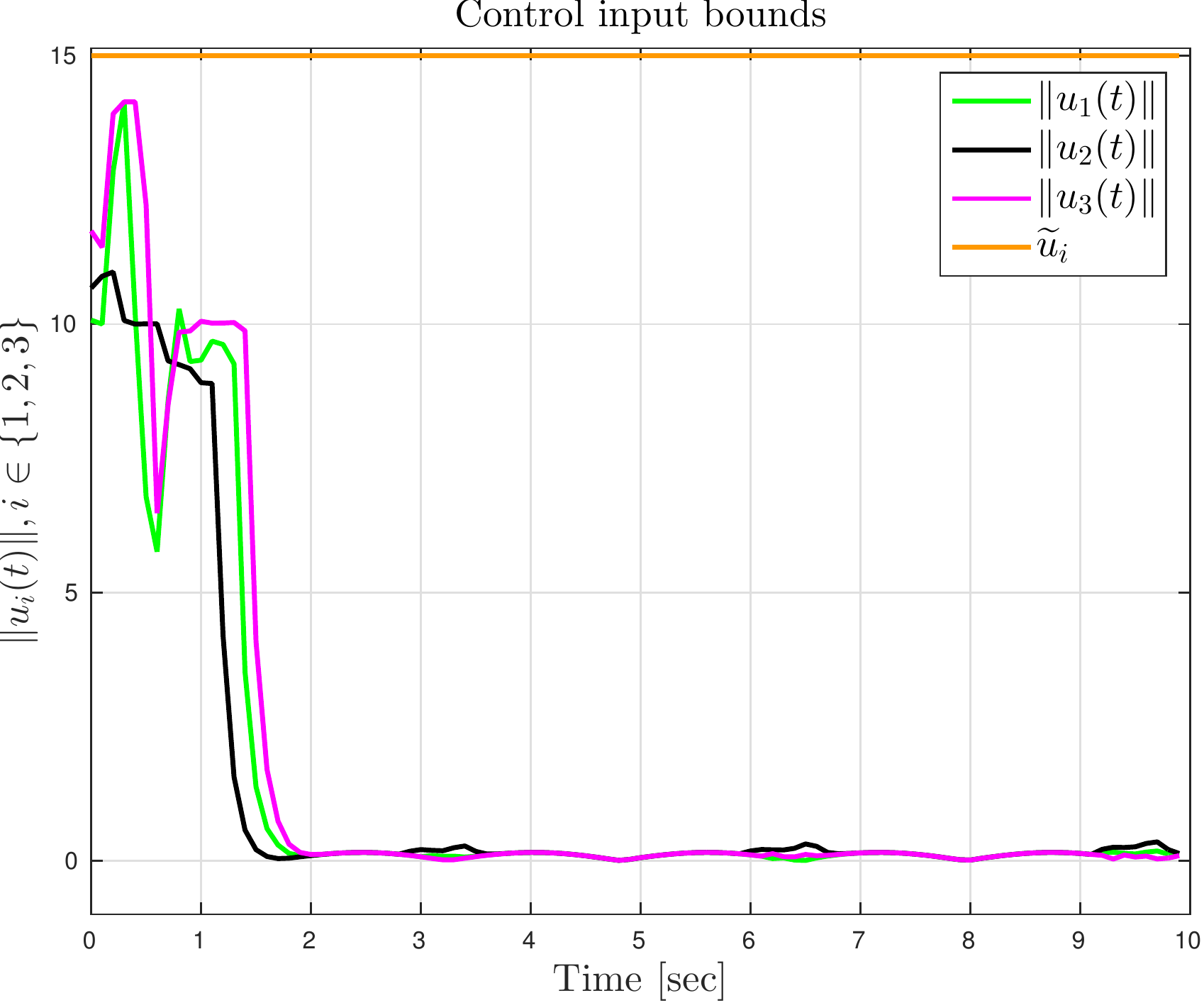}
	\caption{The norms of control inputs signals with $\widetilde{u}_i = 15$.}
	\label{fig:d_ON_res_3_2_inputs_agent_2}
\end{figure}

\section{Conclusions} \label{sec:conclusions}

This paper addresses the problem of stabilizing a multi-agent system under constraints relating to the maintenance of connectivity between
initially connected agents, the aversion of collision among agents and between agents and stationary obstacles within their working environment. Control input saturation as well as disturbances and model uncertainties are also taken into consideration. The proposed control law is a class of Decentralized Nonlinear Model Predictive Controllers. Simulation results verify the controller efficiency of the proposed framework. Future efforts will be devoted to reduce the communication burden between the agents by introducing event-triggered communication control laws. 

\appendices

\section{Proof of Property \ref{property:set_property}} \label{app:proof_set_property}

\noindent Consider the vectors $u$, $v$, $w$, $x \in \mathbb{R}^n$. According to Definition \ref{def:p_difference}, we have that:
\begin{align*}
\mathcal{S}_1 \ominus \mathcal{S}_2 & = \{u \in \mathbb{R}^n: u+v \in \mathcal{S}_1, \forall \ v \in \mathcal{S}_2 \}, \\
\mathcal{S}_2 \ominus \mathcal{S}_3 & = \{w \in \mathbb{R}^n: w+x \in \mathcal{S}_2, \forall \ x \in \mathcal{S}_3 \}.
\end{align*}
Then, by adding the aforementioned sets according to Definition \ref{def:p_difference} we get:
\begin{align}
(\mathcal{S}_1 \ominus \mathcal{S}_2) \oplus (\mathcal{S}_2 \ominus \mathcal{S}_3) & \notag \\
&\hspace{-18mm}= \{u + w \in \mathbb{R}^n : u + v \in \mathcal{S}_1 \ \text{and} \  w + x \in \mathcal{S}_2, \forall \ v \in \mathcal{S}_2, \forall \ x \in \mathcal{S}_3 \} \notag \\
&\hspace{-18mm}= \{u + w \in \mathbb{R}^n : u + v + w + x \in (\mathcal{S}_1 \oplus \mathcal{S}_2), \forall \ v + x \in (\mathcal{S}_2 \oplus \mathcal{S}_3) \}. \label{eq:u_v_w_x}
\end{align}
By setting $s_1 = u + w \in \mathbb{R}^n$, $s_2 = v + x \in \mathbb{R}^n$ and employing Definition \ref{def:p_difference}, \eqref{eq:u_v_w_x} becomes:
\begin{align*}
(\mathcal{S}_1 \ominus \mathcal{S}_2) \oplus (\mathcal{S}_2 \ominus \mathcal{S}_3) & = \{s_1 \in \mathbb{R}^n : s_1 + s_2 \in (\mathcal{S}_1 \oplus \mathcal{S}_2), \forall \ s_2 \in (\mathcal{S}_2 \oplus \mathcal{S}_3) \} \notag \\
&= (\mathcal{S}_1 \oplus \mathcal{S}_2) \ominus (\mathcal{S}_2 \oplus \mathcal{S}_3),
\end{align*}
which concludes the proof. \qed

\section{Proof of Property \ref{property 1}} \label{app:proof_property_2}

\noindent By setting $z \triangleq e + z_{\text{des}} \in \mathbb{R}^n$, $z' \triangleq e' + z_{\text{des}} \in \mathbb{R}^n$ in \eqref{eq:system} we get:
\begin{align*}
\|f_i(e + z_{\text{des}},u)-f_i(e' + z_{\text{des}},u)\| & \le L_{f_i} \|e  + z_{\text{des}}-e' - z_{\text{des}}\|. \notag
\end{align*}
By using \eqref{eq:function_g_i}, the latter becomes:
\begin{align*}
\|g_i(e, u)-g_i(e', u)\| & \le L_{g_i} \|e-e'\|, \notag
\end{align*}
where $L_{g_i} = L_{f_i}$, which leads to the conclusion of the proof. \qed

\section{Proof of Lemma \ref{lemma:F_i_bounded_K_class}} \label{app:proof_lemma_1}

\noindent By invoking the fact that:
\begin{align} \label{eq:rayleigh_inequality}
\lambda_{\min}(P) \|y\|^2 \le y^\top P y \le \lambda_{\max}(P) \|y\|^2, \forall y \in \mathbb{R}^n, P \in \mathbb{R}^{n \times n}, P = P^\top > 0,
\end{align}
we have:
\begin{align*}
e_i^\top Q_i e_i + u_i^\top R_i u_i & \le \lambda_{\max}(_i) \|e_i\|^2 + \lambda_{\max}(R_i) \|u_i\|^2 \notag \\
& = \max \{\lambda_{\max}(Q_i), \lambda_{\max}(R_i) \} \|\eta_i \|^2, \notag
\end{align*}
and:
\begin{align*}
e_i^\top Q_i e_i + u_i^\top R_i u_i & \ge \lambda_{\min}(Q_i) \|e_i\|^2 + \lambda_{\min}(R_i) \|u_i\|^2 \notag \\
& = \min \{\lambda_{\min}(Q_i), \lambda_{\min}(R_i) \} \|\eta_i \|^2, \notag
\end{align*}
where $\eta_i = \left[ e_i^\top, u_i^\top\right]^\top$ and $i \in \mathcal{V}$. Thus, we get:
\begin{align*}
\min \{\lambda_{\min}(Q_i), \lambda_{\min}(R_i) \} \|\eta_i \|^2 & \le e_i^\top Q_i e_i + u_i^\top R_i u_i \le \notag \\
&\hspace{25mm} \max \{\lambda_{\max}(Q_i), \lambda_{\max}(R_i) \} \|\eta_i \|^2.
\end{align*}
By defining the $\mathcal{K}_{\infty}$ functions $\alpha_1$, $\alpha_2 : \mathbb{R}_{\ge 0}  \to \mathbb{R}_{\ge 0}$:
\begin{align*}
\alpha_1(y) \triangleq \min \{\lambda_{\min}(Q_i), \lambda_{\min}(R_i) \} \|y \|^2, \alpha_2(y) \triangleq \max \{\lambda_{\max}(Q_i), \lambda_{\max}(R_i) \} \|y \|^2,
\end{align*}
we get:
\begin{align*}
\alpha_1\big(\|\eta_i\|\big) \leq F_i\big(e_i, u_i\big) \leq \alpha_2\big(\|\eta_i \|\big),
\end{align*}
which leads to the conclusion of the proof. \qed

\section{Proof of Lemma \ref{lemma:F_Lipschitz}} \label{app:proof_of_F_lipsitz}

\noindent For every $e_i, e_i' \in \mathcal{E}_i$, and $u_i \in \mathcal{U}_i$ it holds that:
\begin{align}
\big|F_i(e_i, u_i) - F_i(e'_i, u_i)\big|
&= \big|e_i^{\top} Q_i e_i + u_i^{\top} R_i u_i
-(e_i')^{\top} Q_i e'_i - u_i^{\top} R_i u_i \big| \notag \\
&= \big|e_i^{\top} Q_i (e_i -e'_i) - (e')^{\top} Q_i (e_i-e'_i)\big| \notag \\
&\leq \big|e_i^{\top} Q_i (e_i -e'_i)\big| + \big| (e'_i)^{\top} Q_i (e_i-e'_i)\big|. \label{eq:inequality_lemma_2}
\end{align}
By employing the property that:
\begin{align*}
|e_i^{\top} Q_i e_i'| \leq \|e_i \|  \| Q_i e_i'\| \le  \| Q_i \| \|e_i \|  \| e_i'\| \le \sigma_{\max}(Q_i) \|e_i\| \|e'_i\|,
\end{align*}
\eqref{eq:inequality_lemma_2} is written as:
\begin{align*}
\big|F_i(e_i, u_i) - F_i(e_i', u_i)\big| &\leq \sigma_{\max}(Q_i) \|e_i\| \|e_i - e_i'\| +
\sigma_{\max}(Q_i) \|e_i'\| \|e_i - e_i'\| \\
&\le \sigma_{\max}(Q_i) \sup \limits_{e_i, e_i' \in \mathcal{E}_i}
\left\{ \|e_i\| + \|e_i'\| \right\} \| e_i - e_i'\| \\
&\le \left[ 2 \sigma_{\max}(Q_i) \sup \limits_{e_i \in \mathcal{E}_i} \|e_i\|  \right]  \|e_i - e_i'\| \notag \\
&= L_{F_i}  \|e_i - e_i'\|.  \hspace{85mm} \qed 
\end{align*} 

\section{Proof of Lemma \ref{lemma:diff_state_from_same_conditions}} \label{app:proof_lemma_diff_state}

\noindent By employing Property \ref{remark:predicted_actual_equations_with_disturbance} and substituting $\tau \equiv t_k$ and $s \equiv t_k + s$ in \eqref{eq:remark_4_eq_1}, \eqref{eq:remark_4_eq_2} yields:
\begin{align*}
e_i\big(t_k + s;\ \overline{u}_i\big(\cdot;\ e_i(t_k)\big), e_i(t_k)\big) &= \notag \\
&\hspace{-10mm} e_i(t_k)
+ \int_{t_k}^{t_k + s} g_i\big(e_i(s';\ e_i(t_k)), \overline{u}_i(s')\big) ds'
+ \int_{t_k}^{t_k + s} w_i(\cdot, s')ds', \\
\overline{e}_i\big(t_k + s;\ \overline{u}_i\big(\cdot;\ e_i(t_k)\big), e_i(t_k)\big) &=
e_i(t_k) + \int_{t_k}^{t_k + s} g_i\big(\overline{e}_i(s';\ e_i(t_k)), \overline{u}_i(s')\big) ds',
\end{align*}
respectively. Subtracting the latter from the former and taking norms on both sides yields:
\begin{align*}
& \bigg\|  e_i\big(t_k + s;\ \overline{u}_i\big(\cdot;\ e_i(t_k)\big), e_i(t_k)\big) -
\overline{e}_i\big(t_k + s;\ \overline{u}_i\big(\cdot;\ e_i(t_k)\big), e_i(t_k)\big) \bigg\| \\
&=\bigg\| \int_{t_k}^{t_k + s} g_i\big(e_i(s';\ e_i(t_k)), \overline{u}_i(s')\big) ds'
- \int_{t_k}^{t_k + s} g_i\big(\overline{e}_i(s';\ e_i(t_k)), \overline{u}_i(s')\big) ds' \notag \\
&\hspace{95mm} + \int_{t_k}^{t_k + s} w_i(\cdot, s')ds' \bigg\| \\
&\leq L_{g_i} \int_{t_k}^{t_k + s} \bigg\| e_i\big(s;\ \overline{u}_i\big(\cdot;\ e_i(t)\big), e_i(t)\big) -
\overline{e}_i\big(s;\ \overline{u}_i\big(\cdot;\ e_i(t)\big), e_i(t)\big) \bigg\| ds + s \widetilde{w}_i,
\end{align*}
since, according to Property \ref{property 1}, $g_i$ is Lipschitz continuous in $\mathcal{E}_i$ with Lipschitz constant
$L_{g_i}$. Then, we get:
\begin{align*}
& \bigg\| e_i\big(t_k+s;\ \overline{u}_i\big(\cdot;\ e_i(t_k)\big), e_i(t_k)\big) -
\overline{e}_i\big(t_k+s;\ \overline{u}_i\big(\cdot;\ e_i(t_k)\big), e_i(t_k)\big) \bigg\| \\
& \hspace{-2mm} \leq s \widetilde{w}_i
+ L_{g_i} \int_{0}^{s} \bigg\| e_i\big(t_k + s';\ \overline{u}_i\big(\cdot;\ e_i(t_k)\big), e_i(t_k)\big) -
\overline{e}_i\big(t_k + s';\ \overline{u}_i\big(\cdot;\ e_i(t_k)\big), e_i(t_k)\big) \bigg\| ds'. \\
\end{align*}
By applying the Gr\"{o}nwall-Bellman inequality (see \cite[Appendix A]{khalil_nonlinear_systems}) we get:
\begin{align*}
\bigg\| e_i\big(t_k + s;\ \overline{u}_i\big(\cdot;\ e_i(t_k)\big), e_i(t_k)\big) &-
\overline{e}_i\big(t_k + s;\ \overline{u}_i\big(\cdot;\ e_i(t_k)\big), e_i(t_k)\big) \bigg\| \\
&\leq s \widetilde{w}_i +  L_{g_i} \int_{0}^{s} s' \widetilde{w}_i e^{L_{g_i}(s - s')} ds' \\
&= \dfrac{\widetilde{w}_i}{L_{g_i}} (e^{L_{g_i}s} - 1). \hspace{70mm} \qed
\end{align*}

\section{Proof of Property \ref{property:restricted_constraint_set}} \label{app:proof_of_property_restricted_const}

\noindent Let us define the function $\zeta_i : \mathbb{R}_{\geq 0} \to \mathbb{R}^{n}$ as:
$\zeta_i(s) \triangleq e_i(s) - \overline{e}_i(s;\ u_i(s;\ e_i(t_k)), e_i(t_k))$,
for $s \in [t_k, t_k + T_p]$. According to Lemma \ref{lemma:diff_state_from_same_conditions} we have that:
\begin{align}
\|\zeta_i(s)\| = \|e_i(s) - \overline{e}_i\big(s;\ u_i(s;\ e_i(t)), e_i(t)\big)\|
&\leq \dfrac{\widetilde{w}_i}{L_{g_i}} (e^{L_{g_i} (s-t)} - 1), s \in [t_k, t_k + T_p], \notag
\end{align}
which means that $\zeta_i(s) \in \mathcal{X}_{i,s-t}$. Now we have that
$\overline{e}_i\big( s;\ u_i(\cdot,\ e_i(t_k)), e_i(t_k) \big) \in \mathcal{E}_i \ominus \mathcal{X}_{i,s-t_k}$.
Then, it holds that:
\begin{align*}
\zeta_i(s) + \overline{e}_i\big(s;\ u_i(s;\ e_i(t_k)), e_i(t_k)\big)
&\in \big(\mathcal{E}_i \ominus \mathcal{X}_{i,s-t_k}\big) \oplus \mathcal{X}_{i,s-t_k}.
\end{align*}
or
\begin{align*}
e_i(s) &\in \big(\mathcal{E}_i \ominus \mathcal{X}_{i,s-t_k}\big) \oplus \mathcal{X}_{i,s-t_k}.
\end{align*}
Theorem 2.1 (ii) from \cite{kolmanovsky} states that for every $U,V \subseteq \mathbb{R}^n$ it holds that: $\left(U \ominus V \right) \oplus V \subseteq U$. By invoking the latter result we get:
\begin{align*}
e_i(s) &\in \big(\mathcal{E}_i \ominus \mathcal{X}_{i,s-t_k}\big) \oplus \mathcal{X}_{i,s-t_k} \subseteq \mathcal{E}_i  \Rightarrow e_i(s) \in \mathcal{E}_i, s \in [t_k, t_k+T_p],
\end{align*}
which concludes the proof. \qed

\section{Proof of Lemma \ref{lemma:V_i_lower_upper_bounded}} \label{app:proof_lemma_V_i_lower_boudned}

\noindent By invoking \eqref{eq:rayleigh_inequality} we get:
\begin{align*}
\lambda_{\min}(P_i) \|e_i\|^2 \le e_i^\top P_i e_i \le \lambda_{\max}(P_i) \|e_i\|^2, \forall e_i \in \Psi_i, i \in \mathcal{V}.
\end{align*}
By defining the $\mathcal{K}_{\infty}$ functions $\alpha_1$, $\alpha_2 : \mathbb{R}_{\ge 0}  \to \mathbb{R}_{\ge 0}$:
\begin{align*}
\alpha_1(y) \triangleq \lambda_{\min}(P_i) \|y \|^2, \alpha_2(y) \triangleq \lambda_{\max}(P_i) \|y \|^2,
\end{align*}
we get:
\begin{align*}
\alpha_1\big(\|e_i\|\big) \leq V_i(e_i) \leq \alpha_2\big(\| e_i \|\big), \forall e_i \in \Psi_i, i \in \mathcal{V}.
\end{align*}
which leads to the conclusion of the proof. \qed

\section{Feasibility Analysis} \label{app:feasibility_analysis}

In this section we will show that there can be constructed an admissible but not necessarily optimal control input according to Definition \ref{definition:admissible_input_with_disturbance}.

Consider a sampling instant $t_k$ for which a solution $\overline{u}_i^{\star}\big(\cdot;\ e_i(t_k)\big)$ to Problem $1$ exists.
Suppose now a time instant $t_{k+1}$ such that $t_k < t_{k+1} < t_k + T_p$, and consider that the
optimal control signal calculated at $t_k$ is comprised by the following two
portions:

\begin{equation}
\overline{u}_i^{\star}\big(\cdot;\ e_i(t_k)\big) = \left\{
\begin{array}{ll}
\overline{u}_i^{\star}\big(\tau_1;\ e_i(t_k)\big), & \tau_1 \in [t_k, t_{k+1}], \\
\overline{u}_i^{\star}\big(\tau_2;\ e_i(t_k)\big), & \tau_2 \in [t_{k+1}, t_k + T_p].
\end{array}
\right.
\label{eq:optimal_input_portions_with_disturbances}
\end{equation}

Both portions are admissible since the calculated optimal control input is
admissible, and hence they both conform to the input constraints.
As for the resulting predicted states, they satisfy the state constraints, and,
crucially:
\begin{align}
\overline{e}_i\big(t_k + T_p;\ \overline{u}_i^{\star}(\cdot), e_i(t_k)\big) \in \Omega_i.
\label{eq:predicted_t_k_T_p_from_t_k_in_omega}
\end{align}
Furthermore, according to condition $(3)$ of Theorem \ref{theorem:with_disturbances}, there exists an admissible (and certainly not guaranteed optimal feedback control) input $\kappa_i \in \mathcal{U}_i$ that renders $\Psi_i$ (and consequently $\Omega_i$) invariant over $[t_k + T_p, t_{k+1} + T_p]$.

Given the above facts, we can construct an admissible input $\widetilde{u}_i(\cdot)$  for time $t_{k+1}$ by sewing together the second portion of \eqref{eq:optimal_input_portions_with_disturbances} and the admissible input $\kappa_i(\cdot)$:

\begin{equation}
\widetilde{u}_i(\tau) = \left\{
\begin{array}{ll}
\overline{u}_i^{\star}\big(\tau;\ e_i(t_k)\big), & \tau \in [t_{k+1}, t_k + T_p], \\
\kappa_i \big(\overline{e}_i\big(\tau;\ \overline{u}_i^{\star}(\cdot), e_i(t_{k+1})\big)\big), & \tau \in (t_k + T_p, t_{k+1} + T_p].
\end{array}
\right.
\label{eq:optimal_input_t_plus_one_with_disturbances}
\end{equation}

Applied at time $t_{k+1}$, $\widetilde{u}_i(\tau)$
is an admissible control input with respect to the input constraints as a composition of admissible control inputs, for all $\tau \in [t_{k+1}, t_{k+1} + T_p]$. What remains to prove is the following two statements:\\

\noindent \textbf{Statement 1 :} $e_i\big(t_{k+1} + s;\ \overline{u}_i^{\star}(\cdot), e_i(t_{k+1})\big) \in \mathcal{E}_i$, $\forall s \in [0,T_p]$.

\noindent \textbf{Statement 2 :} $\overline{e}_i\big(t_{k+1} + T_p;\ \widetilde{u}_i(\cdot), e_i(t_{k+1}) \big) \in \Omega_i$. \\

\noindent \textbf{Proof of Statement 1 :} Initially we have that: $\overline{e}_i\big(t_{k+1} + s;\ \widetilde{u}_i(\cdot), e_i(t_{k+1})\big) \in \mathcal{E}_i \ominus \mathcal{X}_{s}$,
for all $s \in [0, T_p]$. By applying Lemma \ref{lemma:diff_state_from_same_conditions} for $t=t_{k+1} + s$ and $\tau=t_k$ we get
\begin{align*}
\bigg\|
e_i\big(t_{k+1}+s;\  \overline{u}_i^{\star}(\cdot), e_i(t_k)\big)
-\overline{e}_i\big(t_{k+1} + s;\ \overline{u}_i^{\star}(\cdot), e_i(t_k)\big)
\bigg\|
\leq \dfrac{\widetilde{w}_i}{L_{g_i}}\big(e^{L_{g_i} (h+s)}-1\big),
\end{align*}
or equivalently:
\begin{align*}
e_i\big(t_{k+1}+s;\  \overline{u}_i^{\star}(\cdot), e_i(t_k)\big)
-\overline{e}_i\big(t_{k+1} + s;\ \overline{u}_i^{\star}(\cdot), e_i(t_k)\big)
\in \mathcal{X}_{i, h+s}.
\end{align*}

By applying a reasoning identical to the proof of Lemma \ref{lemma:diff_state_from_same_conditions} for $t=t_{k+1}$ (in the model equation) and $t = t_k$ (in the real model equation), and $\tau = s$ we get:
\begin{align*}
\bigg\|
e_i\big(t_{k+1}+s;\  \overline{u}_i^{\star}(\cdot), e_i(t_k)\big)
-\overline{e}_i\big(t_{k+1} + s;\ \overline{u}_i^{\star}(\cdot), e_i(t_{k+1})\big)
\bigg\|
\leq \dfrac{\widetilde{w}_i}{L_{g_i}}\big(e^{L_{g_i} s}-1\big),
\end{align*}
which translates to:
\begin{align*}
e_i\big(t_{k+1}+s;\  \overline{u}_i^{\star}(\cdot), e_i(t_k)\big)
-\overline{e}_i\big(t_{k+1} + s;\ \overline{u}_i^{\star}(\cdot), e_i(t_{k+1})\big)
\in \mathcal{X}_{i,s}.
\end{align*}

Furthermore, we know that the solution to the optimization problem is
feasible at time $t_k$, which means that: $\overline{e}_i\big(t_{k+1} + s;\ \overline{u}_i^{\star}(\cdot), e_i(t_k)\big) \in \mathcal{E}_i \ominus \mathcal{X}_{i,h+s}$. Let us for sake of readability set:
\begin{align*}
e_{i, 0} &= e_i\big(t_{k+1} + s;\ \overline{u}_i^{\star}(\cdot), e_i(t_k)\big), \\
\overline{e}_{i, 0} &= \overline{e}_i\big(t_{k+1} + s;\ \overline{u}_i^{\star}(\cdot), e_i(t_k)\big), \\
\overline{e}_{i, 1} &= \overline{e}_i\big(t_{k+1} + s;\ \overline{u}_i^{\star}(\cdot), e_i(t_{k+1})\big),
\end{align*}
and translate the above system of inclusion relations:
\begin{align*}
e_{i,0} - \overline{e}_{i,0} \in \mathcal{X}_{i, h+s}, e_{i,0} - \overline{e}_{i,1} \in \mathcal{X}_{i,s}, \overline{e}_{i,0} \in \mathcal{E}_i \ominus \mathcal{X}_{i,h+s}.
\end{align*}
First we will focus on the first two relations, and we will derive a result
that will combine with the third statement so as to prove that the predicted
state will be feasible from $t_{k+1}$ to $t_{k+1} + T_p$. Subtracting the
second from the first yields
\begin{align*}
\overline{e}_{i,1} - \overline{e}_{i,0} \in \mathcal{X}_{i, h+s} \ominus \mathcal{X}_{i,s}.
\end{align*}
Now we use the third relation $\overline{e}_{i,0} \in \mathcal{E}_i \ominus \mathcal{X}_{i,h+s}$, along with: $\overline{e}_{i,1}- \overline{e}_{i,0} \in \mathcal{X}_{i, h+s} \ominus \mathcal{X}_{i,s}$. Adding the latter to the former yields:
\begin{align*}
\overline{e}_{i,1} \in \big(\mathcal{E}_i \ominus \mathcal{X}_{i,h+s}\big) \oplus \big(\mathcal{X}_{i, h+s} \ominus \mathcal{X}_{i,s}\big).
\end{align*}
By using \eqref{eq:a_minus_b_plus_c} of Property \ref{property:set_property} we get:
\begin{align*}
\overline{e}_{i,1} \in \big(\mathcal{E}_i \oplus \mathcal{X}_{i,h+s}\big) \ominus \big(\mathcal{X}_{i, h+s} \oplus \mathcal{X}_{i,s}\big).
\end{align*}
Using implication \footnote{$A = B_1 \oplus B_2 \Rightarrow A \ominus B = (A \ominus B_1) \ominus B_2$}
(v) of Theorem 2.1 from \cite{kolmanovsky} yields:
\begin{align*}
\overline{e}_{i,1} \in \bigg(\big(\mathcal{E}_i \oplus \mathcal{X}_{i,h+s}\big) \ominus \mathcal{X}_{i, h+s}\bigg) \ominus \mathcal{X}_{i,s}.
\end{align*}
Using implication \footnote{$(A \oplus B) \ominus B \subseteq A$}
(3.1.11) from \cite{schneider_2013} yields
\begin{align*}
\overline{e}_{i,1} \in \mathcal{E}_i \ominus \mathcal{X}_{i,s},
\end{align*}
or equivalently:
\begin{align}
\overline{e}_i\big(t_{k+1} + s;\ \overline{u}_i^{\star}(\cdot), e_i(t_{k+1})\big) \in \mathcal{E}_i \ominus \mathcal{X}_{i,s},\
\forall s \in [0,T_p].
\label{eq:feasibility_2}
\end{align}

By consulting with Property \ref{property:restricted_constraint_set}, this means that the state of the ``true" system does not violate the constraints $\mathcal{E}_i$
over the horizon $[t_{k+1}, t_{k+1} + T_p]$:
\begin{align}
& \overline{e}_i\big(t_{k+1} + s;\ \overline{u}_i^{\star}(\cdot), e_i(t_{k+1})\big) \in \mathcal{E}_i \ominus \mathcal{X}_{i,s} \notag \\
& \Rightarrow\
e_i\big(t_{k+1} + s;\ \overline{u}_i^{\star}(\cdot), e_i(t_{k+1})\big) \in \mathcal{E}_i,\ \forall s \in [0,T_p].
\label{eq:feasibility_with_disturbance_second_point}
\end{align}

\noindent \textbf{Proof of Statement 3}: To prove this statement we begin with:
\begin{align}
V_i\big(\overline{e}_i\big(t_k + T_p;&\ \overline{u}_i^{\star}(\cdot), e_i(t_{k+1})\big)\big)
- V_i\big(\overline{e}_i\big(t_k + T_p;\ \overline{u}_i^{\star}(\cdot), e_i(t_k)\big) \notag \\
&\leq \bigg|  V_i\big(\overline{e}_i\big(t_k + T_p;\ \overline{u}_i^{\star}(\cdot), e_i(t_{k+1})\big)\big)
- V_i\big(\overline{e}_i\big(t_k + T_p;\ \overline{u}_i^{\star}(\cdot), e_i(t_k)\big) \bigg | \notag \\
&\leq L_{V_i}\bigg \| \overline{e}_i\big(t_k + T_p;\ \overline{u}_i^{\star}(\cdot), e_i(t_{k+1})\big)
- \overline{e}_i\big(t_k + T_p;\ \overline{u}_i^{\star}(\cdot), e_i(t_k) \big)\Big\|. \label{eq:from_DV_to_De}
\end{align}
Consulting with Remark \ref{remark:predicted_actual_equations_with_disturbance} we get that the two terms inside the norm are respectively equal to:
\begin{align*}
\overline{e}_i\big(t_k + T_p;\ \overline{u}_i^{\star}(\cdot), e_i(t_{k+1})\big)
&= e_i(t_{k+1}) + \int_{t_{k+1}}^{t_k + T_p} g_i\big(\overline{e}_i(s;\ e_i(t_{k+1})), \overline{u}_i^{\star}(s) \big)ds,
\end{align*}
and
\begin{align*}
& \overline{e}_i\big(t_k + T_p;\ \overline{u}_i^{\star}(\cdot), e_i(t_k)\big)
= e_i(t_k)     + \int_{t_k}^{t_k + T_p} g_i\big(\overline{e}_i(s;\ e_i(t_k)), \overline{u}_i^{\star}(s) \big)ds \\
&= e_i(t_k)     + \int_{t_k}^{t_{k+1}} g_i\big(\overline{e}_i(s;\ e_i(t_k)), \overline{u}_i^{\star}(s) \big)ds
+ \int_{t_{k+1}}^{t_k + T_p} g_i\big(\overline{e}_i(s;\ e_i(t_k)), \overline{u}_i^{\star}(s) \big)ds \\
&= \overline{e}_i(t_{k+1}) + \int_{t_{k+1}}^{t_k + T_p} g_i\big(\overline{e}_i(s;\ e_i(t_k)), \overline{u}_i^{\star}(s) \big)ds.
\end{align*}
Subtracting the latter from the former and taking norms on both sides we get:
\begin{align*}
& \bigg \|\overline{e}_i\big(t_k + T_p;\ \overline{u}_i^{\star}(\cdot), e_i(t_{k+1})\big)
- \overline{e}_i\big(t_k + T_p;\ \overline{u}_i^{\star}(\cdot), e_i(t_k)\big) \bigg\| \\
&= \bigg \| e_i(t_{k+1}) - \overline{e}_i(t_{k+1}) + \int_{t_{k+1}}^{t_k + T_p} g_i\big(\overline{e}_i(s;\ e_i(t_{k+1})), \overline{u}_i^{\star}(s) \big)ds \\
&\hspace{50mm} - \int_{t_{k+1}}^{t_k + T_p} g_i\big(\overline{e}_i(s;\ e_i(t_k)), \overline{u}_i^{\star}(s) \big)ds \bigg \| \\
&\leq \bigg \| e_i(t_{k+1}) - \overline{e}_i(t_{k+1}) \bigg\| + \bigg \| \int_{t_{k+1}}^{t_k + T_p} g_i\big(\overline{e}_i(s;\ e_i(t_{k+1})), \overline{u}_i^{\star}(s) \big)ds \\
&\hspace{50mm} - \int_{t_{k+1}}^{t_k + T_p} g_i\big(\overline{e}_i(s;\ e_i(t_k)), \overline{u}_i^{\star}(s) \big)ds \bigg \| \\
&\leq \bigg \| e_i(t_{k+1}) - \overline{e}_i(t_{k+1}) \bigg\| \\
&+  L_{g_i} \int_{t_{k+1}}^{t_k + T_p} \bigg\| \overline{e}_i\big(s;\ \overline{u}_i^{\star}(\cdot), e_i(t_{k+1})\big)
- \overline{e}_i\big(s;\ \overline{u}_i^{\star}(\cdot), e_i(t_k)\big) \bigg\| ds \\
&= \bigg \| e_i(t_{k+1}) - \overline{e}_i(t_{k+1}) \bigg\| \\
& +  L_{g_i} \int_{h}^{T_p} \bigg\| \overline{e}_i\big(t_k + s;\ \overline{u}_i^{\star}(\cdot), e_i(t_{k+1})\big)
- \overline{e}_i\big(t_k + s;\ \overline{u}_i^{\star}(\cdot), e_i(t_k)\big) \bigg\| ds.
\end{align*}
By applying the  Gr\"{o}nwall-Bellman inequality we obtain:
\begin{align*}
\bigg \|\overline{e}_i\big(t_k + T_p;\ \overline{u}_i^{\star}(\cdot), e_i(t_{k+1})\big)
- \overline{e}_i\big(t_k + T_p;&\ \overline{u}_i^{\star}(\cdot), e_i(t_k)\big) \bigg\| \\
&\leq \bigg \| e_i(t_{k+1}) - \overline{e}_i(t_{k+1}) \bigg\| e^{L_{g_i} (T_p - h)}.
\end{align*}
By applying Lemma \ref{lemma:diff_state_from_same_conditions} for $t = t_k$ and
$\tau = h$ we have:
\begin{align*}
\bigg \|\overline{e}_i\big(t_k + T_p;\ \overline{u}_i^{\star}(\cdot), e_i(t_{k+1})\big)
- \overline{e}_i\big(t_k + T_p;\ \overline{u}_i^{\star}(\cdot), e_i(t_k)\big) \bigg\|
\leq \dfrac{\widetilde{w}_i}{L_{g_i}} (e^{L_{g_i}h} - 1) e^{L_{g_i} (T_p - h)}.
\end{align*}

Hence \eqref{eq:from_DV_to_De} becomes:
\begin{align}
V_i\big(\overline{e}_i\big(&t_k + T_p;\ \overline{u}_i^{\star}(\cdot), e_i(t_{k+1})\big)\big)
- V_i\big(\overline{e}_i\big(t_k + T_p;\ \overline{u}_i^{\star}(\cdot), e_i(t_k)\big) \notag \\
&\leq L_{V_i}\bigg \| \overline{e}_i\big(t_k + T_p;\ \overline{u}_i^{\star}(\cdot), e_i(t_{k+1})\big)
- \overline{e}_i\big(t_k + T_p;\ \overline{u}_i^{\star}(\cdot), e_i(t_k) \big)\Big\| \notag \\
&= L_{V_i} \dfrac{\widetilde{w}_i}{L_{g_i}} (e^{L_{g_i}h} - 1) e^{L_{g_i} (T_p - h)}.
\label{eq:from_DV_to_eq}
\end{align}

Since the solution to the optimization problem is assumed to be feasible at time $t_k$, all states fulfill their respective constraints, and in particular, from \eqref{eq:predicted_t_k_T_p_from_t_k_in_omega}, the predicted state
$\overline{e}_i\big(t_k + T_p;\ \overline{u}_i^{\star}(\cdot), e_i(t_k)\big) \in \Omega_i$.
This means that
$V_i\big(\overline{e}_i\big(t_k + T_p;\ \overline{u}_i^{\star}(\cdot), e_i(t_k)\big) \leq \varepsilon_{\Omega_i}$.
Hence \eqref{eq:from_DV_to_eq} becomes:
\begin{align*}
V_i\big(\overline{e}_i\big(&t_k + T_p;\ \overline{u}_i^{\star}(\cdot), e_i(t_{k+1})\big)\big) \\
&\leq V_i\big(\overline{e}_i\big(t_k + T_p;\ \overline{u}_i^{\star}(\cdot), e_i(t_k)\big) + L_{V_i} \dfrac{\widetilde{w}_i}{L_{g_i}} (e^{L_{g_i}h} - 1) e^{L_{g_i} (T_p - h)} \\
&\leq \varepsilon_{\Omega_i} + L_{V_i} \dfrac{\widetilde{w}_i}{L_{g_i}} (e^{L_{g_i}h} - 1) e^{L_{g_i} (T_p - h)}.
\end{align*}
From Assumption 4 of Theorem \ref{theorem:with_disturbances}, the upper bound of the disturbance is in turn bounded by:
\begin{align*}
\widetilde{w}_i \leq \dfrac{\varepsilon_{\Psi_i} - \varepsilon_{\Omega_i}}{\dfrac{L_{V_i}}{L_{g_i}} (e^{L_{g_i}h} - 1) e^{L_{g_i} (T_p - h)}}.
\end{align*}
Therefore:
\begin{align*}
V_i\big(\overline{e}_i\big(t_k + T_p;\ \overline{u}_i^{\star}(\cdot), e_i(t_{k+1})\big)\big)
\leq  \varepsilon_{\Omega_i} - \varepsilon_{\Omega_i} + \varepsilon_{\Psi_i} = \varepsilon_{\Psi_i}.
\end{align*}
or, expressing the above in terms of $t_{k+1}$ instead of $t_k$:
\begin{align*}
V_i\big(\overline{e}_i\big(t_{k+1} + T_p-h;\ \overline{u}_i^{\star}(\cdot), e_i(t_{k+1})\big)\big) \leq \varepsilon_{\Psi_i}.
\end{align*}
This means that the state
$\overline{e}_i\big(t_{k+1} + T_p-h;\ \overline{u}_i^{\star}(\cdot), e_i(t_{k+1})\big) \in \Psi_i$.
From Assumption \ref{ass:psi_omega}, and since $\Psi_i \subseteq \Phi_i$, there is an admissible control signal
$\kappa_i \big(\overline{e}_i\big(t_{k+1} + T_p-h;\ \overline{u}_i^{\star}(\cdot), e_i(t_{k+1})\big)\big)$
such that:
\begin{align*}
\overline{e}_i\big(t_{k+1} + T_p;\ \kappa_i(\cdot), \overline{e}_i\big(t_{k+1} + T_p - h;\ \overline{u}_i^{\star}(\cdot), e_i(t_{k+1})\big)\big) \in \Omega_i.
\end{align*}
Therefore, overall, it holds that:
\begin{align}
\overline{e}_i\big(t_{k+1} + T_p;\ \widetilde{u}_i(\cdot), e_i(t_{k+1}) \big) \in \Omega_i.
\label{eq:feasibility_3}
\end{align}

Piecing the admissibility of $\widetilde{u}_i(\cdot)$ from \eqref{eq:optimal_input_t_plus_one_with_disturbances} together with conclusions \eqref{eq:feasibility_2} and \eqref{eq:feasibility_3}, we conclude that the
application of the control input $\widetilde{u}_i(\cdot)$ at time
$t_{k+1}$ results in that the states of the real system fulfill their intended constraints during the entire horizon $[t_{k+1}, t_{k+1} + T_p]$. Therefore,
overall, the (sub-optimal) control input $\widetilde{u}_i(\cdot)$ is
admissible at time $t_{k+1}$ according to Definition
\ref{definition:admissible_input_with_disturbance}, which means
that feasibility of a solution to the optimization problem at time $t_k$ implies
feasibility at time $t_{k+1} > t_k$. Thus, since at time $t=0$ a solution is
assumed to be feasible, a solution to the optimal control problem is feasible
for all $t \geq 0$. \qed

\section{Convergence Analysis} \label{app:convergence_analysis}

The second part of the proof involves demonstrating that the state $e_i$ is ultimately bounded in $\Omega_i$. We will show that the
\textit{optimal} cost $J_i^{\star}\big(e_i(t)\big)$ is an ISS Lyapunov function for the closed loop system \eqref{eq:error_system_perturbed}, under the control input \eqref{eq:position_based_optimal_u_2}, according to Definition \ref{def:ISS_Lyapunov}, with: $$J_i^{\star}\big(e_i(t)\big) \triangleq J_i \Big(e_i(t), \overline{u}_i^{\star}\big(\cdot;\ e_i(t)\big)\Big).$$ for notational convenience, let us as  before define the following:
terms:

\begin{itemize}
	\item $u_{0,i}(\tau) \triangleq \overline{u}_i^{\star}\big(\tau;\ e_i(t_k)\big)$
	as the \textit{optimal} input that results from the solution to Problem $1$ based on the measurement of state
	$e_i(t_k)$, applied at time $\tau \geq t_k$;
	\item $e_{0,i}(\tau) \triangleq \overline{e}_i\big(\tau;\ \overline{u}_i^{\star}\big(\cdot;\ e_i(t_k)\big), e_i(t_k)\big)$
	as the \textit{predicted} state at time $\tau \geq t_k$, that is, the predicted state that results from the application of the above input
	$\overline{u}_i^{\star}\big(\cdot;\ e_i(t_k)\big)$ to the
	state $e_i(t_k)$, at time $\tau$;
	\item $u_{1,i}(\tau) \triangleq \widetilde{u}_i(\tau)$
	as the \textit{admissible} input at $\tau \geq t_{k+1}$
	(see \eqref{eq:optimal_input_t_plus_one_with_disturbances});
	\item $e_{1,i}(\tau) \triangleq \overline{e}_i\big(\tau;\ \widetilde{u}_i(\cdot), e_i(t_{k+1})\big)$
	as the \textit{predicted} state at time $\tau \geq t_{k+1}$, that is,
	the predicted state that results from the application of the above input
	$\widetilde{u}_i(\cdot)$ to the state $e_i\big(t_{k+1};$ $\overline{u}_i^{\star}\big(\cdot;\ e_i(t_k)\big), e_i(t_k)\big)$, at time $\tau$.
\end{itemize}

Before beginning to prove convergence, it is worth noting that while the cost
$$J_i \Big(e_i(t), \overline{u}_i^{\star}\big(\cdot;\ e_i(t)\big)\Big),$$
is optimal (in the sense that it is based on the optimal input, which provides
its minimum realization), a cost that is based on a plainly admissible
(and thus, without loss of generality, sub-optimal) input
$u_i \not= \overline{u}_i^{\star}$ will result in a configuration where
\begin{equation*}
J_i \Big(e_i(t), u_i\big(\cdot;\ e_i(t)\big)\Big)
\geq J_i \Big(e_i(t), \overline{u}_i^{\star}\big(\cdot;\ e_i(t)\big)\Big).
\end{equation*}

Let us now begin our investigation on the sign of the difference between the cost
that results from the application of the feasible input $u_{1,i}$,
which we shall denote by $\overline{J}_i\big(e_i(t_{k+1})\big)$,
and the optimal cost $J_i^{\star}\big(e_i(t_k)\big)$, while recalling that:
$J_i \big(e_i(t), \overline{u}_i (\cdot)\big)$ $=$
$\int_{t}^{t + T_p} F_i \big(\overline{e}_i(s), \overline{u}_i (s)\big) ds$ $+$
$V_i \big(\overline{e}_i (t + T_p)\big)$:
\begin{align}
\overline{J}_i\big(e_i(t_{k+1})\big) - J_i^{\star}\big(e_i(t_k)\big) =\
& V_i \big(e_{1,i} (t_{k+1} + T_p)\big) + \int_{t_{k+1}}^{t_{k+1} + T_p} F_i \big(e_{1,i}(s), u_{1,i} (s)\big) ds \notag \\
-&V_i \big(e_{0,i} (t_k + T_p)\big) - \int_{t_k}^{t_k + T_p} F_i \big(e_{0,i}(s), u_{0,i} (s)\big) ds.
\end{align}
Considering that $t_k < t_{k+1} < t_k + T_p < t_{k+1} + T_p$, we break down the
two integrals above in between these integrals:
\begin{align}
&\overline{J}_i\big(e_i(t_{k+1})\big) - J_i^{\star}\big(e_i(t_k)\big) = \notag \\
V_i \big(e_{1,i} (t_{k+1} + T_p)\big)
&\hspace{0mm}+ \int_{t_{k+1}}^{t_k + T_p} F_i \big(e_{1,i}(s), u_{1,i} (s)\big) d s
+ \int_{t_k + T_p}^{t_{k+1} + T_p} F_i \big(e_{1,i}(s), u_{1,i} (s)\big) d s \notag \\
-V_i \big(e_{0,i} (t_k + T_p)\big)
&\hspace{0mm}- \int_{t_k}^{t_{k+1}} F_i \big(e_{0,i}(s), u_{0,i} (s)\big) d s
- \int_{t_{k+1}}^{t_k + T_p} F_i \big(e_{0,i}(s), u_{0,i} (s)\big) d s.
\label{eq:convergence_4_integrals_2}
\end{align}

\noindent Let us first focus on the difference between the two intervals in \eqref{eq:convergence_4_integrals_2} over $[t_{k+1}, t_{k+1} + T_p]$:
\begin{align}
& \int_{t_{k+1}}^{t_k  + T_p} F_i  \big(e_{1,i}(s), u_{1,i} (s)\big) d s
- \int_{t_{k+1}}^{t_k + T_p} F_i \big(e_{0,i}(s), u_{0,i} (s)\big) d s \notag \\
&= \int_{t_k+h}^{t_k + T_p} F_i \big(e_{1,i}(s), u_{1,i} (s)\big) d s
- \int_{t_k+h}^{t_k + T_p} F_i \big(e_{0,i}(s), u_{0,i} (s)\big) d s \notag \\
& \leq \bigg| \int_{t_k+h}^{t_k + T_p} F_i \big(e_{1,i}(s), u_{1,i} (s)\big) d s
- \int_{t_k+h}^{t_k + T_p} F_i \big(e_{0,i}(s), u_{0,i} (s)\big) d s \bigg| \notag \\
& = \bigg| \int_{t_k+h}^{t_k + T_p} \bigg( F_i \big(e_{1,i}(s), u_{1,i} (s)\big)
-  F_i \big(e_{0,i}(s), u_{0,i} (s)\big) \bigg) d s \bigg| \notag \\
& = \int_{t_k+h}^{t_k + T_p} \bigg| F_i \big(e_{1,i}(s), u_{1,i} (s)\big)
-  F_i \big(e_{0,i}(s), u_{0,i} (s)\big) \bigg| d s \notag \\
& \leq L_{F_i}\int_{t_k+h}^{t_k + T_p} \bigg\| \overline{e}_i\big(s;\ u_{1,i} (\cdot), e_i(t_k + h) \big)
-  \overline{e}_i\big(s;\ u_{0,i} (\cdot), e_i(t_k)\big) \bigg\| d s \notag
\end{align}
\begin{align}
& = L_{F_i}\int_{h}^{T_p} \bigg\| \overline{e}_i\big(t_k + s;\ \overline{u}_i^{\star}(\cdot), e_i(t_k + h) \big)
-  \overline{e}_i\big(t_k + s;\ \overline{u}_i^{\star}(\cdot), e_i(t_k)\big) \bigg\| d s.
\label{eq:integrals_over_same_u_LV}
\end{align}
Consulting with Remark \ref{remark:predicted_actual_equations_with_disturbance} for the two different initial conditions we get:
\begin{align*}
\overline{e}_i\big(t_k + s;\ \overline{u}_i^{\star}(\cdot), e_i(t_k +h)\big)
&= e_i(t_k +h) + \int_{t_k +h}^{t_k + s} g_i\big(\overline{e}_i(\tau;\ e_i(t_k + h)), \overline{u}_i^{\star}(\tau)\big) d\tau,
\end{align*}
and
\begin{align}
& \overline{e}_i\big(t_k + s;\ \overline{u}_i^{\star}(\cdot), e_i(t_k)\big) = e_i(t_k) + \int_{t_k}^{t_k + s} g_i\big(\overline{e}_i(\tau;\ e_i(t_k)), \overline{u}_i^{\star}(\tau)\big) d\tau \notag \\
&= e_i(t_k) + \int_{t_k}^{t_k + h} g_i\big(\overline{e}_i(\tau;\ e_i(t_k)), \overline{u}_i^{\star}(\tau)\big) d\tau + \int_{t_k + h}^{t_k + s} g_i\big(\overline{e}_i(\tau;\ e_i(t_k)), \overline{u}_i^{\star}(\tau)\big) d\tau. \notag
\end{align}
Subtracting the latter from the former and taking norms on either side yields
\begin{align}
&\bigg \| \overline{e}_i\big(t_k + s;\ \overline{u}_i^{\star}(\cdot), e_i(t_k +h)\big)
-\overline{e}_i\big(t_k + s;\ \overline{u}_i^{\star}(\cdot), e_i(t_k)\big) \bigg \| \notag \\
&= \bigg\| e_i(t_k +h)
- \bigg( e_i(t_k) + \int_{t_k}^{t_k + h} g_i\big(\overline{e}_i(\tau;\ e_i(t_k)), \overline{u}_i^{\star}(\tau)\big) d\tau \bigg) \notag \\
&+ \int_{t_k +h}^{t_k + s} g_i\big(\overline{e}_i(\tau;\ e_i(t_k + h)), \overline{u}_i^{\star}(\tau)\big) d\tau
- \int_{t_k + h}^{t_k + s} g_i\big(\overline{e}_i(\tau;\ e_i(t_k)), \overline{u}_i^{\star}(\tau)\big) d\tau \bigg\| \notag \\
&= \bigg\| e_i(t_k +h) - \overline{e}_i(t_k + h) \notag \\
&+ \int_{t_k +h}^{t_k + s} \bigg( g_i\big(\overline{e}_i(\tau;\ e_i(t_k + h)), \overline{u}_i^{\star}(\tau)\big)
-  g_i\big(\overline{e}_i(\tau;\ e_i(t_k)), \overline{u}_i^{\star}(\tau)\big) \bigg) d\tau \bigg\| \notag \\
&\leq \bigg\| e_i(t_k +h) - \overline{e}_i(t_k + h) \bigg \| \notag \\
&+ \bigg\| \int_{t_k +h}^{t_k + s} \bigg( g_i\big(\overline{e}_i(\tau;\ e_i(t_k + h)), \overline{u}_i^{\star}(\tau)\big)
-  g_i\big(\overline{e}_i(\tau;\ e_i(t_k)), \overline{u}_i^{\star}(\tau)\big) \bigg) d\tau \bigg\| \notag \\
&\leq \bigg\| e_i(t_k +h) - \overline{e}_i(t_k + h) \bigg \| \notag \\
&+ \int_{t_k +h}^{t_k + s} \bigg\| g_i\big(\overline{e}_i(\tau;\ e_i(t_k + h)), \overline{u}_i^{\star}(\tau)\big)
-  g_i\big(\overline{e}_i(\tau;\ e_i(t_k)), \overline{u}_i^{\star}(\tau)\big) \bigg\| d\tau \notag \\
&\leq \bigg\| e_i(t_k +h) - \overline{e}_i(t_k + h) \bigg \| \notag \\
&+ L_{g_i} \int_{t_k +h}^{t_k + s} \bigg\| \overline{e}_i\big(\tau;\ \overline{u}_i^{\star}(\cdot), e_i(t_k + h) \big)
-  \overline{e}_i\big(\tau;\ \overline{u}_i^{\star}(\cdot), e_i(t_k)\big) \bigg\| d\tau \notag \\
&= \bigg\| e_i(t_k +h) - \overline{e}_i(t_k + h) \bigg \| \notag \\
&+ L_{g_i} \int_{h}^{s} \bigg\| \overline{e}_i\big(t_k + \tau;\ \overline{u}_i^{\star}(\cdot), e_i(t_k + h) \big)
-  \overline{e}_i\big(t_k + \tau;\ \overline{u}_i^{\star}(\cdot), e_i(t_k)\big) \bigg\| d\tau.
\label{eq:df_interim_es}
\end{align}

\noindent By using Lemma \ref{lemma:diff_state_from_same_conditions} and applying the the Gr\"{o}nwall-Bellman inequality, \eqref{eq:df_interim_es} becomes:
\begin{align*}
\bigg \| \overline{e}_i\big(t_k + s;\ \overline{u}_i^{\star}(\cdot), e_i(t_k +h)\big)
&-\overline{e}_i\big(t_k + s;\ \overline{u}_i^{\star}(\cdot), e_i(t_k)\big) \bigg \| \notag \\
&\leq \bigg\| e_i(t_k +h) - \overline{e}_i(t_k + h) \bigg \| e^{L_{g_i}(s-h)} \notag \\
&\leq \dfrac{\widetilde{w}_i}{L_{g_i}} (e^{L_{g_i}h} - 1) e^{L_{g_i}(s-h)}.
\end{align*}

\noindent Given the above result, \eqref{eq:integrals_over_same_u_LV} becomes:
\begin{align*}
\int_{t_{k+1}}^{t_k + T_p} F_i \big(e_{1,i}(s),& u_{1,i} (s)\big) d s
- \int_{t_{k+1}}^{t_k + T_p} F_i \big(e_{0,i}(s), u_{0,i} (s)\big) d s \\
& \leq L_{F_i}\int_{h}^{T_p} \dfrac{\widetilde{w}_i}{L_{g_i}} (e^{L_{g_i}h} - 1) e^{L_{g_i}(s-h)} d s \\
& = L_{F_i} \dfrac{\widetilde{w}_i}{L_{g_i}^2} (e^{L_{g_i} h} - 1) (e^{L_{g_i}(T_p-h)} - 1).
\end{align*}
Hence, we have that:

\begin{align}
\int_{t_{k+1}}^{t_k + T_p} F_i \big(e_{1,i}(s), u_{1,i} (s)\big) d s
&- \int_{t_{k+1}}^{t_k + T_p} F_i \big(e_{0,i}(s), u_{0,i} (s)\big) d s \notag \\
& \leq L_{F_i} \dfrac{\widetilde{w}_i}{L_{g_i}^2} (e^{L_{g_i} h} - 1) (e^{L_{g_i}(T_p-h)} - 1).
\label{eq:end_result_two_integrals}
\end{align}

With this result established, we turn back to the remaining terms found in \eqref{eq:convergence_4_integrals_2} and, in particular, we focus on
the integral
\begin{align*}
\int_{t_k + T_p}^{t_{k+1} + T_p} F_i \big(e_{1,i}(s), u_{1,i} (s)\big) d s.
\end{align*}

We discern that the range of the above integral has a length \footnote{$(t_{k+1} + T_p) - (t_k + T_p) = t_{k+1} - t_k = h$} equal to the length of the interval where \eqref{eq:phi_psi} of Assumption \ref{ass:psi_psi} holds. Integrating \eqref{eq:phi_psi} over the interval $[t_k + T_p, t_{k+1} + T_p]$, for the controls and states applicable in it we get:
\begin{align*}
& \int_{t_k + T_p}^{t_{k+1} + T_p} \Bigg(\dfrac{\partial V_i}{\partial e_{1,i}} g_i\big(e_{1,i}(s), u_{1,i}(s)\big)
+ F_i\big(e_{1,i}(s), u_{1,i}(s)\big)\Bigg) ds \leq 0 \\
& \Leftrightarrow \int_{t_k + T_p}^{t_{k+1} + T_p} \dfrac{d}{ds} V_i\big(e_{1,i}(s)\big) d s
+ \int_{t_k + T_p}^{t_{k+1} + T_p} F_i\big(e_{1,i}(s), u_{1,i}(s)\big) ds \leq 0 \\
& \Leftrightarrow V_i\big(e_{1,i}(t_{k+1} + T_p)\big) - V_i\big(e_{1,i}(t_k + T_p)\big)
+ \int_{t_k + T_p}^{t_{k+1} + T_p} F_i\big(e_{1,i}(s), u_{1,i}(s)\big) ds \leq 0 \\
& \Leftrightarrow V_i\big(e_{1,i}(t_{k+1} + T_p)\big)
+ \int_{t_k + T_p}^{t_{k+1} + T_p} F_i\big(e_{1,i}(s), u_{1,i}(s)\big) ds \leq V_i\big(e_{1,i}(t_k + T_p)\big).
\end{align*}

The left-hand side expression is the same as the first two terms in the right-hand side of equality \eqref{eq:convergence_4_integrals_2}. We can introduce the third one by subtracting it from both sides:
\begin{align*}
V_i\big(e_{1,i}(t_{k+1} + T_p)\big)
&+ \int_{t_k + T_p}^{t_{k+1} + T_p} F_i\big(e_{1,i}(s), u_{1,i}(s)\big) ds
- V_i\big(e_{0,i}(t_k + T_p)\big)
\end{align*}
\begin{align*}
&\leq V_i\big(e_{1,i}(t_k + T_p)\big)
- V_i\big(e_{0,i}(t_k + T_p)\big) \\
&\leq L_{V_i}\bigg \|\overline{e}_i\big(t_k + T_p;\ \overline{u}_i^{\star}(\cdot), e_i(t_{k+1})\big)
- \overline{e}_i\big(t_k + T_p;\ \overline{u}_i^{\star}(\cdot), e_i(t_k) \big)\Big\| \\
&\leq L_{V_i} \dfrac{\widetilde{w}_i}{L_{g_i}} (e^{L_{g_i}h} - 1) e^{L_{g_i} (T_p - h)} \text{ (from \eqref{eq:from_DV_to_eq})}.
\end{align*}
Hence, we obtain:
\begin{align}
V_i\big(e_{1,i}(t_{k+1} + T_p)\big)
& + \int_{t_k + T_p}^{t_{k+1} + T_p} F_i\big(e_{1,i}(s), u_{1,i}(s)\big) ds
- V_i\big(e_{0,i}(t_k + T_p)\big) \notag \\
&\leq L_{V_i}\dfrac{\widetilde{w}_i}{L_{g_i}} (e^{L_{g_i}h} - 1) e^{L_{g_i} (T_p - h)}.
\label{eq:end_result_diff_V_plus_int}
\end{align}
Adding the inequalities \eqref{eq:end_result_two_integrals} and \eqref{eq:end_result_diff_V_plus_int} it is derived that:
\begin{align*}
&\int_{t_{k+1}}^{t_k + T_p} F_i \big(e_{1,i}(s), u_{1,i} (s)\big) d s
- \int_{t_{k+1}}^{t_k + T_p} F_i \big(e_{0,i}(s), u_{0,i} (s)\big) d s \\
&+ V_i\big(e_{1,i}(t_{k+1} + T_p)\big)
+ \int_{t_k + T_p}^{t_{k+1} + T_p} F_i\big(e_{1,i}(s), u_{1,i}(s)\big) ds
- V_i\big(e_{0,i}(t_k + T_p)\big) \\
&\leq L_{F_i} \dfrac{\widetilde{w}_i}{L_{g_i}^2} (e^{L_{g_i} h} - 1) (e^{L_{g_i}(T_p-h)} - 1)
+ L_{V_i}\dfrac{\widetilde{w}_i}{L_{g_i}} (e^{L_{g_i}h} - 1) e^{L_{g_i} (T_p - h)}.
\end{align*}
and therefore \eqref{eq:convergence_4_integrals_2}, by bringing the integral ranging from $t_k$ to $t_{k+1}$ to the left-hand side, becomes:
\begin{align*}
\overline{J}_i\big(e_i(t_{k+1})\big)
- J_i^{\star}\big(e_i(t_k)\big)
&+ \int_{t_k}^{t_{k+1}} F_i \big(e_{0,i}(s), u_{0,i} (s)\big) d s \\
&\hspace{-20mm} \leq L_{F_i} \dfrac{\widetilde{w}_i}{L_{g_i}^2} (e^{L_{g_i} h} - 1) (e^{L_{g_i}(T_p-h)} - 1)
+ L_{V_i}\dfrac{\widetilde{w}_i}{L_{g_i}} (e^{L_{g_i}h} - 1) e^{L_{g_i} (T_p - h)}.
\end{align*}

\noindent By rearranging terms, the cost difference becomes bounded by:
\begin{align*}
\overline{J}_i\big(e_i(t_{k+1})\big) - J_i^{\star}\big(e_i(t_k)\big) \ \le \ \xi_i \widetilde{w}_i -\int_{t_k}^{t_{k+1}} F_i \big(e_{0,i}(s), u_{0,i} (s)\big) d s,
\end{align*}
where:
\begin{align*}
\xi_i \triangleq \dfrac{1}{L_{g_i}} \bigg(e^{L_{g_i}h} - 1\bigg)
\bigg[\big(L_{V_i} + \dfrac{L_{F_i}}{L_{g_i}}\big) \big(e^{L_{g_i}(T_p-h)}-1\big) + L_{V_i} \bigg] > 0.
\end{align*}
and $\xi_i \widetilde{w}_i$ is the contribution of the bounded additive
disturbance $w_i(t)$ to the nominal cost difference; $F_i$ is a positive-definite function as a sum of a positive-definite
$u_i^\top R_i u_i$ and a positive semi-definite function
$e_i^\top Q_i e_i$. If we denote by
$\rho_i \triangleq \lambda_{\min}(Q_i, R_i) \geq 0$ the minimum eigenvalue
between those of matrices $R_i, Q_i$, this means that:
\begin{align*}
F_i \big(e_{0,i}(s), u_{0,i} (s)\big) \geq \rho_i \|e_{0,i}(s)\|^2.
\end{align*}
By integrating the above between the interval of interest $[t_k, t_{k+1}]$ we get:
\begin{align*}
-\int_{t_k}^{t_{k+1}} F_i \big(e_{0,i}(s), u_{0,i} (s)\big)
&\leq -\rho_i \int_{t_k}^{t_{k+1}} \| \overline{e}_i(s;\ \overline{u}_i^{\star}, e_i(t_k)) \|^2 ds.
\end{align*}
This means that the cost difference is upper-bounded by:
\begin{align*}
\overline{J}_i\big(e_i(t_{k+1})\big) - J_i^{\star}\big(e_i(t_k)\big)
&\leq \xi_i \widetilde{w}_i -\rho_i \int_{t_k}^{t_{k+1}} \| \overline{e}_i(s;\ \overline{u}_i^{\star}(\cdot), e_i(t_k)) \|^2 ds,
\end{align*}
and since the cost $\overline{J}_i\big(e_i(t_{k+1})\big)$ is, in general,
sub-optimal: $J_i^{\star}\big(e_i(t_{k+1})\big) - \overline{J}_i\big(e_i(t_{k+1})\big) \leq 0$:
\begin{align}
J_i^{\star}\big(e_i(t_{k+1})\big) - J_i^{\star}\big(e_i(t_k)\big)
\leq \xi_i \widetilde{w}_i - \rho_i \int_{t_k}^{t_{k+1}} \| \overline{e}_i(s;\ \overline{u}_i^{\star}(\cdot), e_i(t_k)) \|^2 ds.
\label{eq:J_opt_between_consecutive_k_2}
\end{align}

Let $\Xi_i(e_i) \triangleq J_i^{\star}(e_i)$. Then, between
consecutive times $t_k$ and $t_{k+1}$ when the FHOCP is solved, the above
inequality reforms into:
\begin{align}
\Xi_i\big(e_i(t_{k+1})\big) - \Xi_i\big(e_i(t_k)\big)
&\leq \xi_i \widetilde{w}_i
- \rho_i \int_{t_k}^{t_{k+1}} \| \overline{e}_i(s;\ \overline{u}_i^{\star}(\cdot), e_i(t_k)) \|^2 ds \notag \\
&\leq \int_{t_k}^{t_{k+1}} \bigg( \dfrac{\xi_i}{h} \|w_i(s)\|
-  \rho_i \| \overline{e}_i(s;\ \overline{u}_i^{\star}(\cdot), e_i(t_k)) \|^2 \bigg) ds.
\label{eq:check_for_ISS_here_main_branch}
\end{align}

The functions $\sigma\big(\|w_i\|\big) \triangleq \dfrac{\xi_i}{h} \|w_i\|$ and $\alpha_3\big(\|e_i\|\big) \triangleq \rho_i \|e_i\|^2$ are class
$\mathcal{K}$ functions according to Definition \ref{def:k_class}, and therefore, according to Lemma \ref{lemma:V_i_lower_upper_bounded} and Definition \ref{def:ISS_Lyapunov}, $\Xi_i\big(e_i\big)$ is an ISS Lyapunov function in $\mathcal{E}_i$.

Given this fact, Theorem \ref{def:ISS_Lyapunov_admit_theorem}, implies that
the closed-loop system is input-to-state stable in $\mathcal{E}_i$. Inevitably then,
given Assumptions \ref{ass:psi_psi} and \ref{ass:psi_omega}, and condition $(3)$ of Theorem \ref{theorem:with_disturbances}, the closed-loop
trajectories for the error state of agent $i \in \mathcal{V}$ reach
the terminal set $\Omega_i$ for all
$w_i(t)$ with $\|w_i(t)\| \leq \widetilde{w}_i$, at some point $t = t^{\star} \geq 0$. Once inside
$\Omega_i$, the trajectory is trapped there because of the
implications\footnote{For more details, refer to the discussion after the
	declaration of Theorem $7.6$ in \cite{marquez2003nonlinear}.} of
\eqref{eq:check_for_ISS_here_main_branch} and Assumption \ref{ass:psi_omega}.

In turn, this means that the system \eqref{eq:error_system_perturbed} converges
to $x_{i, \text{des}}$ and is trapped in a vicinity of it $-$ smaller than that
in which it would have been trapped (if actively trapped at all)
in the case of unattenuated disturbances $-$, while simultaneously
conforming to all constraints $\mathcal{Z}_{i}$. This conclusion holds
for all $i \in \mathcal{V}$, and hence, the overall system of agents
$\mathcal{V}$ is stable. \qed

\bibliography{references}
\bibliographystyle{ieeetr}
\end{document}